\documentclass[aps,pra,twocolumn,superscriptaddress,tightenlines,showpacs,longbibliography]{revtex4-1}

\usepackage[T1]{fontenc}
\usepackage{lmodern} 
\usepackage{graphicx,color}
\usepackage{latexsym,amsmath,amssymb,amsfonts,amsthm,mathrsfs,bm,bbm}
\usepackage{enumerate}

\usepackage{verbatim}
\usepackage{epstopdf}
\usepackage[titletoc,title]{appendix}
\usepackage{tikz}

\theoremstyle{plain}
\newtheorem{thm}{Theorem}[section]
\newtheorem{prop}[thm]{Proposition}
\newtheorem{coro}[thm]{Corollary}

\theoremstyle{definition}
\newtheorem{defn}[thm]{Definition}

\theoremstyle{remark}
\newtheorem{remark}[thm]{Remark}

\newcommand{\tr}[2]{\text{Tr}_{#1}\!\left(#2\right)}

\newcommand{\ket}[1]{| #1 \rangle}
\newcommand{\bra}[1]{\langle #1 |}

\newcommand{\ketbra}[1]{\ket{#1}\bra{#1}}

\newcommand{\dimension}[1]{\text{dim}(#1)}

\def\N{\mathcal{N}}
\def\H{\mathcal{H}}
\def\L{\mathcal{L}}

\def\pure{|\psi\rangle}
\def\qure{|\Psi\rangle}
\def\dpure{|\psi\rangle\langle\psi|}

\newcommand{\supp}{{\rm supp}}

\newcommand{\Span}{{\rm span}}

\begin{document}

\title{Uniquely determined pure quantum states need not be unique ground states \\ 
of quasi-local Hamiltonians} 

\author{Salini Karuvade}
\affiliation{\mbox{Department of Physics and Astronomy, Dartmouth 
College, 6127 Wilder Laboratory, Hanover, NH 03755, USA}}
\affiliation{Institute for Quantum Science and Technology, University of Calgary, 2500 University Drive NW, Calgary, Alberta, Canada T2N 1N4}
		
\author{Peter D. Johnson}
\affiliation{Department of Chemistry and Chemical Biology Harvard University, 12 Oxford Street, Cambridge, MA 02138, USA}
\affiliation{Zapata Computing, Inc., 501 Massachusetts Avenue,
Cambridge, MA 02139, USA}

\author{Francesco Ticozzi}
\affiliation{Dipartimento di Ingegneria dell'Informazione,
Universit\`a di Padova, via Gradenigo 6/B, 35131 Padova, Italy} 
\affiliation{\mbox{Department of Physics and Astronomy, Dartmouth 
College, 6127 Wilder Laboratory, Hanover, NH 03755, USA}}

\author{Lorenza Viola}
\affiliation{\mbox{Department of Physics and Astronomy, Dartmouth 
College, 6127 Wilder Laboratory, Hanover, NH 03755, USA}}

\begin{abstract}
We consider the problem of characterizing states of a multipartite quantum system from restricted, quasi-local information, with emphasis on uniquely determined pure states. By leveraging tools from dissipative quantum control theory, we show how the search for states consistent with an assigned list of reduced density matrices may be restricted to a proper subspace, which is determined solely by their supports. The existence of a quasi-local observable which attains its unique minimum over such a subspace further provides a {\em sufficient} criterion for a pure state to be uniquely determined by its reduced states. While the condition that a pure state is uniquely determined is necessary for it to arise as a non-degenerate ground state of a quasi-local Hamiltonian, we prove the opposite implication to be {\em false} in general, by exhibiting an explicit analytic counterexample. We show how the problem of determining whether a quasi-local parent Hamiltonian admitting a given pure state as its unique ground state is dual, in the sense of semidefinite programming, to the one of determining whether such a state is uniquely determined by the quasi-local information. Failure of this dual program to attain its optimal value is what prevents these two classes of states to coincide.
\end{abstract}


\date{\today}
\maketitle 

\section{Introduction}
\label{sec:intro}

Understanding the way in which states of subsystems (``parts'') relate to states of the system as a ``whole'' has contributed to elucidate some of the most profound differences from the classical setting since the early days of quantum mechanics \cite{Schro}, and has since remained a major theme of investigation in quantum information science.  Even in situations where the constituent subsystems represent distinguishable degrees of freedom, so that no additional constraints from quantum statistics arise, and the global quantum state may be assumed to be pure, the relationship between parts and whole remains highly non-trivial in general due to the presence of multipartite entanglement. Given access to local information, as provided by a collection of {\em reduced density matrices} (RDMs) describing subsystem states, the very problem of deciding their ``consistency'' -- namely, the existence of a valid global state, pure or mixed, whose reduced states match the input list -- is known to be QMA-complete \cite{YiKai2006}. In this work, we will assume that such a ``quantum marginal problem'' does admit a solution, and focus on the following inter-related questions: If the {\em joining set} of quantum states that share a specified list of RDMs is indeed non-empty, how can it be characterized and computed?  What conditions can guarantee that it consists of a single element, making the corresponding quantum state {\em uniquely determined among all} (UDA) possible states? What special physical significance do UDA states enjoy, in particular, can every UDA be seen as the {\em unique ground state} (UGS) of some ``locally determined'' Hamiltonian?

Elucidating the above questions has both fundamental and practical implications. On the one hand, characterizing what possible set(s) of RDMs may uniquely determine the underlying quantum state sheds light on how multipartite entanglement is distributed across different subsystems. For instance, failure of a state to be UDA by knowledge of all the RDMs of $k$ out of $N$ total parties ($2\leq k \leq N$) signifies the existence of {\em irreducible} $(k+1)$-party correlations \cite{Linden2002,Linden2002_2,Zhou2008}; among $N$-qubit pure states, most are completely determined by only two of their RDMs of just over half the parties \cite{Linden2002_2,Diosi2004,Jones2005} (in fact, knowledge of their support suffices \cite{SaliniJPA}), whereas the irreducible $N$-party correlation is nonzero only for Greenberger-Horne-Zeilinger (GHZ) type pure states \cite{Linden2002,Walck2009}. On the other hand, determining whether a state of interest is UDA relative to a specified set of RDMs that reflects a certain physical or operational constraint or, more generally, characterizing its joining set, may be crucial to enable or inform quantum tasks ranging from quantum state tomography using local data \cite{Xin2017} to quantum self-testing~\cite{Acin2008}. 

In this work, we provide a deeper understanding of the UDA property for pure states of multipartite quantum systems from a twofold standpoint. Our first contribution is to bring tools from dissipative quantum control -- in particular, quantum state stabilization under (quasi-)locality constraints \cite{Ticozzi2012,TicozziQIC2014} -- to bear on the problem of characterizing the joining set of a target quantum state $\pure$ relative to a specified list of RDMs. While the only general approach known to date entails an extensive search for compatible states over all possible density operators, we show how the search space may be exactly reduced to a specific subspace of the multipartite Hilbert space, called the \emph{Dissipatively Quasi-Locally Stabilizable (DQLS) subspace}, which is determined by the supports of the given RDMs. In addition to substantially lowering the complexity of the underlying search problem, the DQLS subspace may offer an analytical means for establishing the UDA nature of $\pure$ in principle; in particular, we show how it leads to {\em sufficient} conditions for $\pure$ to be UDA via the identification of an appropriate {\em UDA witness} -- namely, an observable whose expectation must be {\em uniquely} extremized over the DQLS subspace. 

Since QL constraints are naturally obeyed by ``few-body'' ($k$-local) Hamiltonians, it is natural to ask whether {\em any} pure state that is UDA must arise as the UGS of a Hamiltonian that respects the same constraint. As our main contribution, we prove that, perhaps surprisingly, this need {\em not} be true in general, by constructing an example of a six-qubit pure state that is UDA by its two-body RDMs yet {\em cannot} be the UGS of any two-body Hamiltonian. This answers, in the negative, a special relevant instance of a broader question posed by Chen and coworkers in \cite{Chen2012}: while they could show that the condition of a space $V$ to be ``$k$-correlated'' -- loosely speaking, to only support states that are uniquely determined by their $k$-body RDMs -- is necessary for $V$ to be a ground-state manifold of a $k$-local Hamiltonian, sufficiency was left open in general. By focusing on the limiting case where the $k$-correlated space $V$ is one-dimensional, we further establish that the problems of determining whether $\pure$ is UDA or, respectively, UGS may be cast as a {\em primal-dual pair} in the language of {\em semidefinite programming} (SDP) \cite{Boyd}. We show that while the condition of $\pure$ being UDA is necessary and sufficient for it being the UGS of a QL Hamiltonian when both the primal and dual programs attain their optimal value, no guarantee exists for this to happen in general. Specifically, we identify the failure of the dual program to achieve its maximum as the mechanism that prevents the equivalence between the UDA and UGS properties to hold in general.

The content is organized as follows. Sec. \ref{sec:prelim} introduces relevant notation and collects definitions and prior results that will be needed in our subsequent analysis. In Sec. \ref{sec:joining_set}, we establish the above-mentioned characterization of the joining set of a multipartite quantum state in terms of the DQLS subspace, by also providing a constructive procedure for reconstructing global states from local data solely based on knowledge of the {\em support} of relevant RDMs.  The concept of a {\em UDA witness} is introduced as well, and the usefulness of the proposed tools is demonstrated by showing how all $N$-qubit states that are equivalent to W states under stochastic local operations assisted by classical communication (SLOCC) are UDA relative to {\em arbitrary} non-trivial QL constraints. This recovers and unifies a number of specific results in the literature \cite{Rana,Wu2014}, by also offering a significantly more streamlined proof technique, applicable to other multipartite pure states of interest in principle. Sec. \ref{sec:equivalence} is devoted to presenting our explicit counterexample of a pure state that is analytically proved to be UDA -- via construction of an appropriate witness -- but not UGS -- by also leveraging dynamical symmetrization ideas \cite{LV99,PZ}. In Sec. \ref{sec:sdp} we cast and analyze the UDA-UGS problem within a general SDP framework, by also contrasting QL with more general linear constraints.  A brief summary along with a discussion of remaining open questions conclude in Sec. \ref{sec:end}, wheareas Appendix \ref{proofs} and Appendix \ref{sec:SDP} are devoted to presenting complete technical proofs and some basic elements of SDP, respectively. 


\section{Notation and Preliminaries}
\label{sec:prelim}

Throughout this work, we consider a multipartite quantum system consisting of $N$ distinguishable, finite-dimensional subsystems. The associated multipartite Hilbert space has a tensor product structure given by 
\[ \H = \bigotimes_{a=1}^N \H_a, {\mathrm{dim}}\,(\H_a)=d_a, \,
{\mathrm{dim}}(\H) =\prod_{a=1}^N d_a\equiv D <\infty.\]
We shall use $\mathcal{B}(\H)$ to denote the $D^2$-dimensional space of (bounded) linear operators acting on $\H$. Density operators representing the physical states of the quantum system are trace-one, positive semi-definite operators which belong to the convex space $\mathcal{D}(\H) \subset \mathcal{B}(\H)$. The extreme points of $\mathcal{D}(\H)$ correspond to the one-dimensional projectors, $\rho=\rho^2 \equiv |\psi\rangle\langle\psi|$, which are in one-to-one correspondence with state vectors $|\psi\rangle\in \H$.  

Multipartite qubit systems will have a special prominence in our discussion.  In this case, $d_a\equiv 2$, for all $a$, and $\H = \otimes_{a=1}^N \H_q \simeq ({\mathbb C}^2)^{\otimes N}$, with $\H_q =\text{span}\{ |0\rangle, |1\rangle\}$. In addition to the standard representation of $N$-qubit pure states as superposition of computational basis states \cite{Nielsen-Chuang}, a more compact notation will also be used, in terms of the ``excitations'' they contain. A qubit is said to be excited if it is in its $|1\rangle$ state.  A product state of $N$-qubits is then described by a string consisting of the indexes of the excited qubits.  For instance, 
\begin{equation}
|0\rangle_N \equiv |0\rangle^{\otimes N},\quad |2\rangle_3\equiv |010\rangle, \quad  |13\rangle_4  \equiv |1010\rangle, 
\label{compact}
\end{equation}
where the subscript in this new notation represents the total number of qubits. Given the above representation for basis states, arbitrary pure states on $N$ qubits are written as a linear combination of different excitation states. Multiqubit operators in $\mathcal{B}(\H)$ will be constructed, as usual, out of a product operator basis consisting of $\{\mathbbm{1}, \sigma^x, \sigma^y, \sigma^z\}$, where $\sigma^\alpha$ are the Pauli operators on ${\mathbb C}^2$, and we will also let $\sigma^\pm \equiv (\sigma^x \pm i \sigma^y)/2$ be raising and lowering spin-angular momentum operators. The notation $\sigma_j^\alpha$ will be used to denote the Pauli operator $\sigma^\alpha$ acting on the $j$-th qubit, that is, $\sigma_j^\alpha \equiv \mathbbm{1}\otimes\cdots\sigma^x\otimes\cdots\mathbbm{1}$, and similarly for ladder operators. 

Single-qubit Pauli operators such the ones above are the simplest example of {\em local} (or uni-local) operators. While constraints imposed from either (or both) the coupling topology and the geometry of the underlying lattice typically restrict the structure of naturally occurring or engineered many-body Hamiltonians, strictly local constraints are too restrictive in practice.  Following prior work \cite{Ticozzi2012,TicozziQIC2014,Johnson2016}, we focus on \emph{Quasi-Local (QL)} constraints, which we formalize as follows:
\begin{defn}
A {\em neighborhood} $\N_k$ is a collection of subsystem indexes given by $\N_k \subsetneq \{1,\dots, N \}$.
A {\em neighborhood structure} $\N \equiv \{\N_k \}_{k=1}^M$ is a finite collection of such neighborhoods. 
\end{defn} 
In this work, we only consider neighborhood structures that are {\em non-trivial}, that is, we require each subsystem $j$ to belong to at least some neighborhood $\N_k$ and each $\N_k$ to overlap with at least another neighborhood $\N_{k'}$.  Notice that the above QL notion includes $k$-local interactions as considered for instance in \cite{Chen2012}; in such a case, the relevant $k$-local neighborhood structure consists of neighborhoods each containing at most $k$ subsystems, with $1<k<N$. Fig.~\ref{fig:chain} illustrates the QL notion that is associated to a one-dimensional lattice (a spin chain) with two-body nearest-neighbor (NN) interactions that let us define a non-trivial neighborhood structure. 

\begin{figure}[t]
\includegraphics[width=\columnwidth]{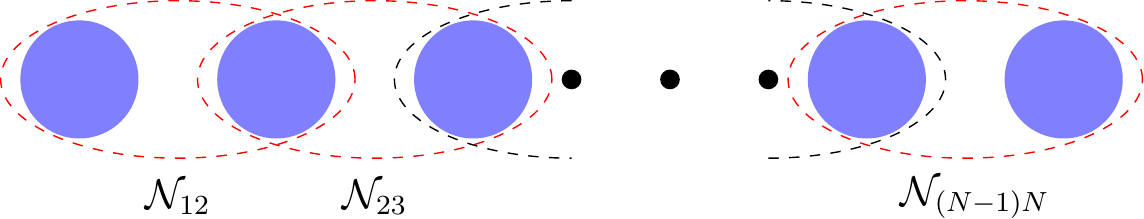}
\vspace*{-4mm}
\caption{(Color online) $N$-party spin chain with two-body NN interactions and open boundary conditions. Dotted ovals represent the neighborhoods. The neighborhood structure is given by $\N \equiv \N_{\mathrm{NN}} = \{\N_{12},\dots,\N_{(N-1)N}\}$. For periodic boundary condition, a neighborhood $\N_{1 N}$ should be included in $\N_{\mathrm{NN}} $ in addition to the above.}
\label{fig:chain}
\end{figure}

Consider a quantum state $ \rho \in \mathcal{D}(\H)$. Given a neighborhood $\N_k$, the RDM of $\rho$ on $\N_k$ is determined by $\rho_{\N_k} \equiv \tr{\overline{\N}_k}{\rho}$, where $\overline{\N}_k$ is the neighborhood complement of $\N_k$, that is, $\N_k \cup \overline{\N}_k = \{1,\dots,N \}$. While the mapping from $\rho$ to the list of RDMs $\{\rho_{\N_k},\, \forall\, \N_k \in \N\}$ is unique for fixed $\N$, the converse is generally not true. This motivates the following definition:
\begin{defn}
The \emph{joining set} of a state $\rho  \in \mathcal{D}(\H)$ relative to $\N= \{\N_k \}_{k=1}^M$, denoted $\mathcal{M}_\mathcal{N}(\rho)$, is the set of all quantum states that share the same list of RDMs on each neighborhood in $\N$ as $\rho$: 
\begin{equation*}
\mathcal{M}_\mathcal{N}(\rho) = \{\sigma \in  \mathcal{D}(\H) : \tr{\overline{\N}_k}{\rho} =  \tr{\overline{\N}_k}{\sigma}, \, \forall\, \N_k \in \N\}.
\end{equation*}
\end{defn}
\noindent
It is immediate to see that $\mathcal{M}_\mathcal{N}(\rho)$ is closed and convex. Determining whether, for a given set of RDMs on $\N$, the corresponding joining set is non-empty, that is, the RDMs are consistent with a valid (pure or mixed) quantum state, is an instance of the quantum marginal problem. The latter is known to be QMA-complete \cite{YiKai2006}.

Quantum states for which the joining set relative to a fixed neighborhood structure consists of a single element deserve special attention:
\begin{defn}
\label{def:UDA}
A quantum state $\rho\in\mathcal{D}(\H)$ is \emph{Uniquely Determined among All states (UDA)}, 
with respect to a neighborhood structure $\N = \{\N_k \}_{k=1}^M$, if there exists no other state $\sigma \in\mathcal{D}(\H)$ 
with the same set of RDMs. 
\end{defn}
\noindent
In other words, $\rho$ is UDA with respect to $\N$ if and only if  $\mathcal{M}_\mathcal{N}(\rho) = \{\rho\}$. A related notion is that of \emph{Uniquely Determined among all Pure states (UDP)} \cite{Xin2017}: A pure state $\rho = \dpure$ is UDP with respect to $\N$ if it is the only pure state belonging to $\mathcal{M}_{\N}(\rho)$.  

A useful construct that helps in understanding the structure of the joining set of a multipartite pure state of interest relative to a specified  neighborhood structure is the so-called \emph{Dissipatively Quasi-Locally Stabilizable (DQLS) subspace}. This concept was introduced in the context of QL stabilization problems in \cite{Ticozzi2012} and was further used and characterized in subsequent related analysis \cite{TicozziQIC2014,Johnson2016,TicozziAlternating,SaliniJPA}. Formally, we have the following:
\begin{defn}  
\label{def:DQLSsubspace}
The \emph{DQLS subspace} of a pure state $\pure \in \H$ relative to the neighborhood structure 
$\N = \{\N_k \}_{k=1}^M$, denoted $\H_{\N}(\pure)$, is given by:
\begin{equation}
\label{eq:DQLSsubspace}
\H_{\N}(\pure) \equiv \bigcap_{\N_k\in\N}\supp(\rho_{\N_k}\otimes I_{\overline{\N}_k}).
\end{equation}  
\end{defn}

\noindent 
From a control-theoretic standpoint, $\H_{\N}(\pure)$ represents the {\em smallest} subspace which contains $\pure$ and can be stabilized via (continuous or discrete-time) Markovian dynamics that is QL relative to $\N$. In particular, in the continuous-time setting, the Lindblad generator is purely dissipative (has vanishing Hamiltonian contribution) in its standard form. 
The DQLS subspace enjoys a number of properties that will be relevant to our analysis:

\smallskip

${\bm{(i)}}$ For given $\pure$ and $\N$, we have 
\begin{equation}
\label{eq:DQLSstate}
\H_{\N}(\pure) = \Span\{\pure \}
\end{equation}
if and only if $\pure$ is the {\em unique ground state} (UGS) of a frustration-free parent Hamiltonian $H$ that is QL relative to $\N$~\cite{Ticozzi2012,Johnson2016}.  That is, we may express $H$ as a sum of neighborhood operators, 
$H\equiv \sum_k H_k= \sum_k H_{\N_k}\otimes I_{\overline{\N}_k}$ and, in addition, the ground space of $H$ is contained in the ground-state space of {\em each} such $H_k$. Accordingly, if Eq. \eqref{eq:DQLSstate} holds, the DQLS subspace is one-dimensional. 

\smallskip

${\bm{(ii)}}$ More generally, if two pure states $\pure$ and $|\phi\rangle$ are equivalent under SLOCC transformations, their DQLS subspaces, relative to any non-trivial neighborhood structure $\N$, have the same dimensionality \cite{SaliniJPA}. 

\smallskip

${\bm{(iii)}}$ Furthermore, as established in \cite{Johnson2016,SaliniJPA}, we have 
\begin{equation}
\supp(\sigma) \subseteq \H_{\N}(\pure), \quad \forall \,\sigma \in \mathcal{M}_{\N}(\pure).
\label{eq:supp0}
\end{equation}
That is, all states in the joining set of a pure state $\rho=\pure\langle\psi|$ relative to $\N$ have their support contained in the corresponding DQLS subspace. 

\smallskip

As we will now show, the DQLS subspace may be leveraged to gain insight about the joining set of general, mixed quantum states as well.

\section{Characterizing the joining set of a multipartite quantum state}
\label{sec:joining_set}

Given a quantum state and a neighborhood structure of interest, no efficient procedure is available to determine whether the state is UDA or construct its joining set in general. In specific scenarios it is possible to make headway by exploiting structural or symmetry properties of the state: for instance, the extent to which {\em generic} pure states in $N$-party systems are UDA by appropriate subsets of RDMs has been extensively investigated \cite{Linden2002_2,Jones2005,Diosi2004,SaliniJPA}; likewise, conclusions have been reached on a case-by-case basis for entangled states encompassing $N$-qubit GHZ-states and their local unitary equivalents \cite{Walck2009}, stabilizer and Dicke states~\cite{Wu20151}, along with generalized W states as we will also further discuss in Sec. \ref{sec:Wstate}.

In principle, the most straightforward way to characterize the joining set $\mathcal{M}_{\N}(\rho)$ of $\rho$ relative to $\N$ is to numerically search over all the possibilities in the associated density operator space $\mathcal{D}(\H)$. However, such a brute-force procedure becomes exponentially harder with increasing $N$. For particular instances this search space may be reduced to a certain extent. For example, if $\pure$ is known to be UDP relative to $\N$, the joining set $\mathcal{M}_{\N}(\pure)$ may be computed by searching {\em only} in the set of genuinely mixed states. Similarly, under the promise that a given set of RDMs in $\N$ corresponds to a pure state that is UDP, one can reconstruct such a global state by searching over the set of pure states alone. Still, such prior information is not always available. Here, we show how, by leveraging the DQLS subspace concept, we can both simplify the problem of characterizing $\mathcal{M}_{\N}(\rho)$ for arbitrary $\rho$ and $\N$ by reducing the search space, and obtain sufficient conditions for a pure state to be UDA.

\subsection{Reduced search space and UDA witnesses}
\label{sub:uda_general} 

As a first step, we extend Definition \ref{def:DQLSsubspace} to general quantum states: the {\em DQLS subspace of $\rho\in {\mathcal D}(\H)$ relative to $\N$} which, with slight abuse of notation, we shall still denote as $\H_{\N}(\rho)$, is given by
\begin{equation}
\H_{\N}(\rho) \equiv  \bigcap_{\N_k\in\N}\supp(\rho_{\N_k}\otimes I_{\overline{\N}_k} ). 
\label{eq:subs}
\end{equation}
\noindent 
Through the following theorem, we show that the desired ``reduced search space'' for UDA characterization is indeed related to the above DQLS subspace. 
\begin{thm}
\label{th:DQLSsubspace}
Let $\rho \in \mathcal{D}(\H)$, with associated joining set $\mathcal{M}_{\N}(\rho)$ relative to a non-trivial neighborhood structure $\N = \{\N_k \}_{k=1}^M$. Then every state in the joining set has support on the corresponding DQLS subspace:
\begin{equation}
\supp(\sigma) \subseteq \H_{\N}(\rho), \quad \forall \,\sigma \in \mathcal{M}_{\N}(\rho).
\label{eq:supp1}
\end{equation}
\end{thm}

\begin{proof}
Since all states in $\mathcal{M}_{\N}(\rho)$ share the same set of RDMs on $\N$, their DQLS subspaces coincide. 
That is, $\H_{\N}(\rho) = \H_{\N}(\sigma)$, for all $\sigma \in \mathcal{M}_{\N}(\rho)$. 
	
To complete the proof it only remains to show that $\supp(\rho) \subseteq \H_{\N}(\rho)$, for arbitrary given $\rho \in \mathcal{D}(\H)$ and $\N$. Fix a neighborhood $\N_k\in\N$. Let $\rho_{\N_k} = \sum_q p_q\Pi_q$ be the spectral decomposition of the RDM $\rho_{\N_k}$, where $\{\Pi_q \}$ is a resolution of the identity. By definition of partial trace, $p_q = \tr{}{\rho\Pi_q\otimes I_{\overline{\N}_k}}$. If $p_{\bar{q}} = 0$ for some ${\bar{q}}$, this implies $\tr{}{\rho\Pi_{\bar{q}}\otimes I_{\overline{\N}_k}} = 0$. Consequently, $\rho(\Pi_{\bar{q}}\otimes I_{\overline{\N}_k}) = 0$ due to the positivity of the two operators. Thus, $\supp(\rho) \perp \supp(\Pi_{\bar{q}}\otimes I_{\overline{\N}_k})$, for all $p_{\bar{q}} = 0$, and $\supp(\rho) \subseteq \cup_{p_q\neq 0}\supp(\Pi_q\otimes I_{\overline{\N}_k}) = \supp(\rho_{\N_k}\otimes I_{\overline{\N}_k})$. This implies 
$\supp(\rho) \subseteq \H_{\N}(\rho)$, as desired.
\end{proof}

Based on Theorem~\ref{th:DQLSsubspace}, an algorithmic procedure for reconstructing global states from local data may be summarized as follows:
\begin{enumerate}
\item[{\bf{(1)}}] Collect information about the \emph{support} of the RDMs on selected neighborhoods $\{\N_k \}$.
\item[{\bf{(2)}}] Construct the associated DQLS subspace, using Eq. \eqref{eq:subs}.
\item[{\bf{(3)}}] Search for states consistent with the given RDMs having {\em support in the DQLS subspace}. 
\end{enumerate}

\noindent 
An immediate corollary also follows: 
\begin{coro}
\label{coro:dim}
Consider $\rho \in \mathcal{D}(\H)$ and a non-trivial neighborhood structure $\N$. Then the rank of any 
state in the joining set of $\rho$ is no larger than the dimensionality of the corresponding DQLS subspace:
$$ {\rm rank}(\sigma) \leq {\rm dim} [\H_{\N} (\rho)] , \quad \forall \sigma \in {\mathcal M}_\N(\rho).$$
\end{coro}
 
\noindent 
As we noted in Eq. \eqref{eq:DQLSstate}, pure states that are UGS of a QL Hamiltonians are associated to a one-dimensional DQLS subspace. By the above corollary, we then recover the known fact that UGS are always UDA relative to the corresponding neighborhood structure \cite{Chen2012}, as we also further discuss in Sec.~\ref{sub:equiv0}. 

Besides reducing the search space for characterizing the joining set, the DQLS subspace is also instrumental in obtaining a sufficient criterion for certifying the UDA property. This is formalized in the 
following: 
\begin{thm}
\label{th:UDAwitness}
Consider a pure state $\pure \in \H$ with DQLS subspace $\H_{\N}(\pure)$ relative to a non-trivial neighborhood structure $\N=\{\N_k\}$. 
Assume that there exists a QL Hermitian operator, ${\mathbb W}= \sum_k {\mathbb W}_{\N_k} \otimes I_{\overline{\N}_k}$, such that 
\begin{equation}
\label{eq:UDAwitness}
\langle \psi | {\mathbb W}  \pure < \langle \phi| {\mathbb W}   |\phi\rangle,
\end{equation}
for any normalized state $|\phi\rangle\in\H_{\N}(\pure)$, with $\ketbra{\phi}\neq\ketbra{\psi}$. Then $\pure$ is UDA relative to $\N$.  
\end{thm}
\begin{proof}
We prove this by contradiction. Assume that $\ket{\psi}$ is not UDA relative to $\N$. Then, there exists some $\sigma\neq\ketbra{\psi}$ with the same marginals as $\ket{\psi}$ in $\N$. By Theorem \ref{th:DQLSsubspace}, $\textup{supp}(\sigma)\subseteq\H_{\N}(\pure)$. That is, $\sigma = \sum_i p_i|\phi_i\rangle\langle \phi_i|$, where $|\phi_i\rangle\in\H_{\N}(\pure)$ and $p_i\geq0$ for all $i$, with $\sum_i p_i =1$.

Also assume that an operator ${\mathbb W}$ obeying Eq.~\eqref{eq:UDAwitness} exists. Owing to the QL nature of ${\mathbb W}$ relative to $\N$,  $\langle \psi|{\mathbb W}\pure = \tr{}{{\mathbb W}\sigma}$. Then in order for the expectation $\tr{}{{\mathbb W}\sigma} = \sum_i p_i \langle \phi_i|{\mathbb W}|\phi_i\rangle$ to be minimum, each pure-state expectation $\langle \phi_i |{\mathbb W} |\phi_i\rangle$ must be minimum.  But, following Eq.~\eqref{eq:UDAwitness},  this is only possible if  $p_1 = 1$ and $|\phi_1\rangle = \pure$, whereby the conclusion follows. \end{proof}

We call a QL observable ${\mathbb W}$ with the above properties a ($\N$-)\emph{UDA witness} for $\pure$. We note that ${\mathbb W}$ only needs to be uniquely {\em extremized} by $\pure$ over the DQLS subspace $\H_{\N}(\pure)$: if $\langle \psi| {\mathbb W} \pure$ is a unique maximum on this subspace, it can still be brought under the purview of Eq.~\eqref{eq:UDAwitness} by simply choosing $-{\mathbb W}$ as the UDA witness. It is important to appreciate that the restriction to a {\em pure} state $\pure$ is crucial. A UDA witness cannot be used, as in Theorem \ref{th:UDAwitness}, to diagnose a proper \emph{mixed} state as UDA. To see this, consider a UDA proper mixed state with spectral decomposition $\rho = \sum_i p_i |\psi_i\rangle \langle \psi_i|$. From Theorem~\ref{th:DQLSsubspace}, each $|\psi_i\rangle$ is in $\H_\N(\pure)$. Since $\rho$ is a proper convex combination of  $\{\ketbra{\psi_i}\}$, there must be at least one pure state, not equal to $\rho$, for which $\tr{}{\mathbb{W}\ketbra{\psi_i}}\leq\tr{}{\mathbb{W}\rho}$.
Thus, $\rho$ is not the unique minimizer of $\tr{}{\mathbb{W}\rho}$ among states with support in $\H_\N(\pure)$ \cite{Terhal}.

Clearly, in order for the DQLS subspace to be useful in analyzing UDA properties it is both important that it is efficiently computable and of sufficiently low dimensionality compared to the full $\H$. We refer to Theorem 4.1 in~\cite{SaliniJPA} for a discussion of conditions under which $\H_{\N}(\pure)= \H$ for a generic multipartite pure state $\pure$. While the complexity of obtaining $\H_{\N}(\rho)$ has not been investigated in general so far, it is often possible to leverage structural properties of the state of interest to analytically characterize $\H_{\N}(\rho)$ for arbitrary $\N$, as we explicitly illustrate next.

\subsection{Application to W state and SLOCC equivalents}
\label{sec:Wstate}

We focus on the paradigmatic W state and its SLOCC class on $N$ qubits. The W state on $N$ qubits is a symmetric combination of one qubit excitations, namely, 
\begin{eqnarray*} 
|{\rm W}\rangle_N &\! \equiv \!& \frac{1}{\sqrt{N}} \Big(|10\ldots0\rangle+|01\ldots 0\rangle +
\ldots |00\ldots1\rangle\Big) \\ 
&\!=\! &\frac{1}{\sqrt{N}}\sum_{k=1}^N|k\rangle_N. 
\end{eqnarray*}
It has been shown that the SLOCC equivalents of $|{\rm W}\rangle_N$ may be represented, up to local unitary 
transformations, by using a {\em generalized W state} \cite{Rana} of the form
\begin{equation}
\label{eq:GW}
|{\rm GW}\rangle_N = c_0|0\rangle_N+\!\sum_{k=1}^Nc_k|k\rangle_N \equiv c_0|0\rangle_N+
\sqrt{1-c_0^2} \, |\overline{\rm W}\rangle_N,
\end{equation}
for real coefficients that obey $c_0 \geq 0$, $c_k >0$, for all $k>1$, $\sum_{k=0}^N c_k^2 = 1$ and where, by construction, $|\overline{\rm W}\rangle_N$ is orthogonal to $|0\rangle_N$. If two multipartite states are related by a local unitary transformation, that is, if 
\[\rho = (\otimes_{k=1}^N U_k)\,\sigma \,(\otimes_{k=1}^N U_k^\dagger),\]
it is immediate to verify that their joining sets relative to any given neighborhood structure $\N$ are also related by the same transformation. Hence we can construct the joining set of the SLOCC class of the W state by investigating the one 
for $|{\rm GW}\rangle_N$.

Let us first characterize the DQLS subspace of $|{\rm GW}\rangle_N$ relative to arbitrary $\N$. Thanks to the fact that, 
as recalled in Sec. \ref{sec:prelim}, the dimension of the DQLS subspace is invariant under SLOCC transformations, we have that 
\begin{equation}
\label{eq:DQLS_SLOCC}
\dimension{\H_{\N}(|{\rm GW}\rangle_N)} = \dimension{\H_{\N}(|{\rm W}\rangle_N)}.
\end{equation}
We can use this result to characterize $\H_{\N}(|{\rm GW}\rangle_N)$: 

\begin{prop}
\label{pr:GW_DQLS}
The DQLS subspace of the generalized W state $|{\rm GW}\rangle_N$ relative to any non-trivial neighborhood structure 
$\N$ is given by
\begin{equation}
\label{eq:GW_DQLS}
\H_{\N}(|{\rm GW}\rangle_N) = \Span\{|0\rangle_N,|\overline{{\rm W}}\rangle_N \}, 
\end{equation}
where the corresponding states are defined in Eq.~\eqref{eq:GW}.
\end{prop}
\begin{proof}
From Theorem~\ref{th:DQLSsubspace}, we know that $|{\rm GW}\rangle_N \in\H_{\N}(|{\rm GW}\rangle_N) $. 
Since we also know from previous analysis \cite{Ticozzi2012,TicozziQIC2014} that the DQLS subspace of the
W state is two-dimensional for any non-trivial $\N$, we only need to find another pure state that is linearly 
independent from $|{\rm GW}\rangle_N$ to fully characterize this subspace. We now show that $|0\rangle_N$ 
also belongs to $\H_{\N}(|{\rm GW}\rangle_N) $.

Fix a neighborhood $\N_j = \{k_1,\dots,k_L\} \in \N$ such that it contains $L<N$ qubits. With respect to the $\N_j\vert \overline{\N}_j$ bi-partition, we can rewrite the one-excitation terms in Eq.~\eqref{eq:GW} in the following form:
\begin{align*}
|k\rangle_N = |f(k)\rangle_L|0\rangle_{N-L},\quad \forall\, k\in\N_j,\\
|k\rangle_N = |0\rangle_L|g(k)\rangle_{N-L},\quad \forall\, k\notin\N_j.
\end{align*}
Here, $ |f(k)\rangle_L,|0\rangle_{L} \in \H_{\N_j}$ and $ |g(k)\rangle_{N-L},|0\rangle_{N-L}\in \H_{\overline{\N}_j}$, and 
$f(k) \in \{1, \ldots, L\} $, $g(k) \in \{1, \ldots, N-L\}$ denote the relative position of the $k$-th qubit that is excited, depending on the neighborhood to which it belongs. In terms of this bi-partition, we can express $|{\rm GW}\rangle_N$ as:
\begin{eqnarray*}
|{\rm GW}\rangle_N  & =&  c_0|0\rangle_N+\sum_{k\in\N_j}c_k|f(k)\rangle_L|0\rangle_{N-L}\\
&+& \sum_{k\notin\N_j}c_k |0\rangle_L|g(k)\rangle_{N-L}. 
\end{eqnarray*}
Now define $|\nu\rangle_L \equiv \sum_{k\in\N_j}c_k|f(k)\rangle_L$.
The RDM of the state in $\N_j$ is given by
\[
\rho_{\N_j} = \big( 1-\langle\nu|\nu\rangle\big)|0\rangle_L\langle 0|+ |\nu\rangle_L\langle\nu|+c_0\big(|0\rangle_L\langle\nu|+|\nu\rangle_L\langle 0|\big).
\] 
It is thus easy to verify that $|0\rangle_L\in\supp(\rho_{\N_j})$ for all $\N_j\in\N$ and, therefore,
\[|0\rangle_N\in\bigcap_{\N_j}\supp(\rho_{\N_j}\otimes I_{\overline{\N}_j}) = \H_{\N}(|{\rm GW}\rangle_N). \]
We conclude that both $|0\rangle_N,|{\rm GW} \rangle_N \in \H_{\N}(|{\rm GW}\rangle_N)$. Choosing an orthonormal basis 
yields Eq.~\eqref{eq:GW_DQLS}.
\end{proof}

Note that determining $\H_{\N}(|{\rm W}\rangle_N)$ is a special case of the above proposition, corresponding to $c_0=0$ and $c_k=1/\sqrt{N}$ for all $k$ in Eq.~\eqref{eq:GW}. This recovers the observation that $\H_{\N}(|{\rm W}\rangle_N) = \Span\{|0\rangle_N,|{\rm W}\rangle_N\}$ originally made in~\cite{Ticozzi2012}. It is remarkable that the search space for determining the joining set of any generalized W state is reduced from the $2^N$-dimensional Hilbert space to the {\em two-dimensional} DQLS subspace, for arbitrary $N$ and $\N$. 

Furthermore, we now exploit the structure of the DQLS subspace of generalized W states to prove that these pure states and, consequently, the entire SLOCC equivalence class of the W state, are indeed UDA relative to arbitrary $\N$. We do so by constructing an explicit UDA witness for the representative state $|{\rm GW}\rangle_N$.  An alternative direct proof that uses the structure of the DQLS subspace itself is also included in Appendix \ref{proofs}.

\begin{coro}
\label{th:GW_UDA}
SLOCC equivalent states of the $N$-qubit W state are UDA relative to any non-trivial neighborhood structure $\N$, with the UDA witness given by
\begin{equation}
\label{eq:witness_GW}
{\mathbb W} \!= \!\left[ \frac{1-2c_0^2}{d_1^2} \Big((d_1^2-1)\mathbbm{1}+\sigma^z_1\Big)+\frac{2c_0}{d_1} \sqrt{1-c_0^2}\,
\sigma^x_1\right]\otimes\mathbbm{1}_{2,N},
\end{equation} 
where $d_1 = c_1 (1-c_0^2)^{-1/2}$ and $ \mathbbm{1}_{2,N}$ is the identity operator acting on qubits $2,\dots,N$.
\end{coro}

\begin{proof}
Although the structure of the QL operator ${\mathbb W}$ may look complicated, we show that the particular choice of UDA witness is not arbitrary. Consider the isometric embedding $V: \H\to \H_2$, where $\H_2  \equiv \Span\{|0\rangle,|1\rangle\}$  a two-dimensional Hilbert space such that 
\[V|0\rangle_N = |0\rangle,\quad V|{\rm W}\rangle_N = |1\rangle, \] 
and $V^\dagger V$ is a projector onto $\H_{\N}(|{\rm GW}\rangle_N)$. We are required to find an Hermitian operator ${\mathbb W}\in\mathcal{B}(\H)$ such that it is QL relative to $\N$ and 
\begin{equation}
\label{eq:GW_max}
{}_N\langle {\rm GW}| {\mathbb W} |{\rm GW}\rangle_N > \langle \Phi| {\mathbb W} |\Phi\rangle,\quad 
\forall\, |\Phi\rangle\in\H_{\N}(|{\rm GW}\rangle_N). 
\end{equation}
Since $V$ is an isometry from $\H_{\N}(|{\rm GW}\rangle_N)$ to $\H_2$, it is distance-preserving. 
Thus, the above expression is equivalent to finding an operator $V {\mathbb W} V^\dagger \in \mathcal{B}(\H_2)$ such that 	
\[\langle {\rm gw}|V \,{\mathbb W} \,V^\dagger|{\rm gw}\rangle > 
\langle \phi|V \,{\mathbb W}\, V^\dagger|\phi\rangle,\, |\phi\rangle\in\H_2,\]
where $|{\rm gw}\rangle \equiv V|{\rm GW}\rangle_N = c_0|0\rangle+\sqrt{1-c_0^2} \,|1\rangle$ and $|\phi\rangle \equiv V|\Phi\rangle$. We now use the Bloch sphere representation of the former \cite{Nielsen-Chuang}, namely, 
$|{\rm gw}\rangle = \cos(\theta/2) |0\rangle+e^{i\phi}\sin(\theta/2)|1\rangle$, 
with $\theta = 2\arccos c_0$ and $\phi = 0$, to observe that 
\begin{equation}
\label{eq:n_sigma}
\vec{n}\cdot\vec{\sigma} \equiv \cos \theta Z + \sin\theta X = (1-2c_0^2)\,Z+2c_0\sqrt{1-c_0^2}\,X,
\end{equation}
is an operator in $\mathcal{B}(\H_2)$, which is uniquely maximized by $|{\rm gw}\rangle$. Our remaining task is to find ${\mathbb W}\in\mathcal{B}(\H)$, such that it is QL relative to $\N$ and $V{\mathbb W}V^\dagger = \vec{n}\cdot\vec{\sigma}$.
	
The QL operator $\sigma^z_1\otimes \mathbbm{1}_{2,N}$ transforms under the isometry $V$ in the following way:
\begin{eqnarray*}
V(\sigma^z_1\otimes \mathbbm{1}_{2N})V^\dagger & = & \begin{pmatrix}
{}_N\langle 0|\sigma_1^z|0\rangle_N
& {}_N\langle 0|\sigma_1^z|{\rm W}\rangle_N\\ 
{}_N\langle {\rm W}|\sigma_1^z|0\rangle_N & {}_N\langle {\rm W}|\sigma_1^z|{\rm W}\rangle_N\end{pmatrix}\\ 
& = &\begin{pmatrix}
	1 & 0\\ 0 & 1-2d_1^2
\end{pmatrix}, 
\end{eqnarray*}
where we have suppressed $\sigma_z^1\otimes \mathbbm{1}_{2,N}$ to $\sigma_z^1$ within the matrix representation. In terms of Pauli operators in $\H_2$, $V(\sigma_z^1\otimes \mathbbm{1}_{2N})V^\dagger = (1-d_1^2)\mathbbm{1}+d_1^2 Z$. We can do a similar analysis for $\sigma^x_1\otimes \mathbbm{1}_{2,N}, \sigma^y_1\otimes \mathbbm{1}_{2,N}$ as well. We now express Pauli operators in $\H_2$ in terms of $\{\sigma^x_1, \sigma^y_1,\sigma^z_1 \}$ and $\mathbbm{1}_N$, the identity operator in $\H$. That is,
\begin{eqnarray}
X & =& \frac{1}{d_1} V(\sigma^x_1\otimes \mathbbm{1}_{2,N})V^\dagger,  \label{eq:Pauli_isometryx}\\
Y & = & \frac{1}{d_1} V (\sigma^y_1\otimes \mathbbm{1}_{2,N})V^\dagger,  \label{eq:Pauli_isometryy}\\
Z &=& \frac{1}{d_1^2}V \left(\sigma^z_1\otimes \mathbbm{1}_{2,N} -(1-d_1^2) \mathbbm{1}_N\right)V^\dagger. \label{eq:Pauli_isometryz}
\end{eqnarray}	
Finally, we combine Eq.~\eqref{eq:witness_GW} and Eqs.~\eqref{eq:Pauli_isometryx}--\eqref{eq:Pauli_isometryz} to verify that $V\,{\mathbb W}\,V^\dagger$ is indeed equal to the operator $\vec{n}\cdot\vec{\sigma}$ in Eq.~\eqref{eq:n_sigma}.  Thus, ${\mathbb W}$ is uniquely maximized by the state $|{\rm GW}\rangle_N$ in $\H_{\N}(|{\rm GW}\rangle_N)$. Since ${\mathbb W}$ is {\em strictly local}, it can serve as the UDA witness for $|{\rm GW}\rangle_N$ relative to any non-trivial neighborhood structure $\N$, as claimed.
\end{proof}

The UDA nature of the $N$-qubit W state and its SLOCC equivalents relative to {\em specific} neighborhood structures has been previously investigated in the literature \cite{Rana,Wu20151}. Aside from being technically simpler and more transparent, our analysis has the advantage of being directly applicable to {\em any} neighborhood structure, as long as it is non-trivial. Thus, the existing results are obtained as special instances of a unified framework. Our approach can be extended to other quantum states as well, provided that the relevant DQLS subspaces can be characterized efficiently. 

\begin{remark}
It is worth spelling out in more detail the connection between Theorem~\ref{th:GW_UDA} and the results reported in \cite{Wu2014}. There, the authors show that the SLOCC class of the $N$-qubit W state is UDA relative to a {\em special} neighborhood structure $\N_{\rm tree}$, which is a collection of $(N-1)$ two-body neighborhoods. The neighborhoods in $\N_{\rm tree}$ are chosen in such a way that the corresponding qubits form a tree-graph on $N$ vertices, whose $(N-1)$ edges represent the relevant two-body neighborhoods. 
	
We now show how the result in \cite{Wu2014} can be also used to arrive to the same conclusions, provided that a suitable procedure is preliminarily implemented in order to reduce an arbitrary non-trivial $\N$ to $\N_{\rm tree}$. Consider the set of RDMs of a quantum state $\rho$ on $\N$, given by $\textbf{R}_{\N} \equiv \{\rho_{\N_k}:\N_k\in\N \}$. By partial trace over the appropriate indexes, we can further produce another set of RDMs, $\textbf{R}_{2} \equiv \{\rho_{ij}:i,j\in\N_k,\,  \forall\N_k \}$, which is the  collection of all two-body RDMs of $\rho$ that can be inferred from $\textbf{R}_{\N}$. Further, $\textbf{R}_2$ may be reduced to a set of $(N-1)$ element $\textbf{R}_{\rm tree} \subseteq \textbf{R}_2$ by only retaining those $\rho_{ij}$s whose indexes constitute the vertices of a tree-graph. Such a reduction is always possible because $\N$ is a non-trivial neighborhood structure. Each neighborhood in $\N$ represents a collection of edges that belong to a connected graph with $N$ vertices. Such a graph is necessarily spanned by a tree-graph containing $(N-1)$ edges, each of which represent a two-body RDM in the list $\textbf{R}_{\rm tree}$. Notice that the construction is not unique in general, as different spanning tree graphs can be considered.
		
Having exhibited a way to reduce $\textbf{R}_{\N}$ to $\textbf{R}_{\rm tree}$, it is easy to see that if two states $\rho,\sigma$ share the same set of RDMs in $\N$, they share the same two-body RDM-list $\textbf{R}_{\rm tree}$ as well. Therefore, any state that is non-UDA relative to a non-trivial $\N$ is also non-UDA relative to $\N_{\rm tree}$. Taking the converse of this statement, since the authors of~\cite{Wu2014} proved that the SLOCC class of the W state is UDA relative to $\N_{\rm tree}$-type neighborhood structures, it follows that they are also UDA relative to any non-trivial $\N$. While this establishes the equivalence between the two results on a formal level, a main advantage of our approach is that it is directly expressed in terms of the original QL constraint of the problem, as desirable from both a physical and a control-engineering perspective.
\end{remark}


\section{UDA pure states as ground states of QL Hamiltonians}
\label{sec:equivalence}

So far we have approached the problem of characterizing the joining set of a quantum state, relative to an arbitrary neighborhood structure, as a search problem in the associated space of density operators. While the DQLS subspace provides insight into the structural features of quantum states that make them UDA or not, a complete mathematical characterization is still lacking. Physically, it also remains unclear how strong a constraint UDA is for a quantum state in a given multipartite setting and, consequently, whether UDA states may commonly occur in typical scenarios. Since most physical Hamiltonians are QL relative to some neighborhood structure, it is then natural to explore what connection may exist between a pure state being a UGS of such a QL Hamiltonian and being UDA with respect to the same constraint. 

\subsection{UGS of QL Hamiltonians are always UDA}
\label{sub:equiv0}

It is easy to verify that pure states that arise as UGS of QL Hamiltonians are necessarily UDA relative to the same neighborhood structure. While this result may be seen as a special instance of the fact that the ground-state space of a $k$-local Hamiltonian is $k$-correlated \cite{Chen2012},  and while we also re-obtained it as a consequence of corollary \ref{coro:dim}, a direct proof is also straightforward:

\begin{prop}
\label{pr:ugsuda}
If $\pure$ is the UGS of a QL Hamiltonian relative to $\N$, then $\pure$ is UDA by its $\N$-neighborhood RDMs.
\end{prop}
\begin{proof}
Let $\pure$ be the UGS of $H = \sum_k H_{\N_k} \otimes I_{\overline{\N}_k}$. Assume that $\ket{\psi}$ is not UDA. Then 
there exists some $\sigma \neq \dpure \in \mathcal{M}_{\N}(\pure)$. Accordingly, $\tr{}{H\dpure} = \tr{}{H\sigma}$, since the energy expectation value of the quantum state depends only on its RDMs in $\N$ due to $H$ being quasi-local with respect to $\N$. It follows that $\sigma$ also belongs to the ground-state space of $H$, contradicting the assumption that $\pure$ is the unique ground state of $H$.  
\end{proof}

Prior to our work, it was not conclusively established whether the UDA property would {\em suffice}, for arbitrary $\pure$, to guarantee the existence of a parent QL Hamiltonian having $\pure$ as its UGS, relative to the given $\N$. Unlike the implication discussed above, it is hard to validate the sufficiency criterion due to the lack of any apparent mathematical connection in this direction \cite{Chen2012,Huber2018}.  Nonetheless, a number of physically relevant UDA pure states are known to arise as the UGS of QL Hamiltonians. For example, the W state on $N$-qubits, which we proved to be UDA relative to arbitrary non-trivial neighborhood structures in the previous section, is also known to be the UGS of a simple XX-type ferromagnetic Hamiltonian in a transverse magnetic field, at least for NN interactions \cite{Bruss2005}. Likewise, injective matrix product states that are DQLS~\cite{Ticozzi2012}, and hence UDA relative to the corresponding neighborhood structure, have also been proved to be the UGS of frustration-free QL Hamiltonians~\cite{Perez-Garcia2007}. Prompted by these positive examples, it is indeed tempting to conjecture the two properties of UDA and UGS to be {\em equivalent} \cite{GuhnePRL}.  We now show, however, that this is {\em false} in general, by exhibiting an explicit counter-example of a pure state that is provably UDA but not UGS.


\subsection{UDA states need not be UGS \\ of QL Hamiltonians}
\label{sec:nonUDA}

Our counter-example involves a six-qubit pure state, which we denote by $\qure_6$ and which, using the compact notation discussed in Sec.~\ref{sec:Wstate}, has the following form:
\begin{equation}
\label{eq:d}
\qure_6 \equiv \frac{1}{\sqrt{2}} \, (|0\rangle_6+|\overline{\text{D}}\rangle_6),
\end{equation}
where $|\overline{\text{D}}\rangle_6$ is a ``modified'' two-excitation Dicke state, with all the NN excitations removed, that is, 
\begin{eqnarray}
\label{eq:barD}
|\overline{\text{D}}\rangle_6 =
\frac{1}{3} ( \, & | & 13\rangle_6+ |14\rangle_6+|15\rangle_6+|24\rangle_6+|25\rangle_6 \nonumber \\
                +\,  & | & 26\rangle_6+|35\rangle_6+|36\rangle_6+|46\rangle_6).
\end{eqnarray}
Accordingly, we can rewrite $\qure_6$ as
\begin{equation}
\qure_6 = \frac{1}{\sqrt{2}}\Big(|0\rangle_6+\frac{1}{3}\!\sum_{\,(j-i) >1}|ij\rangle_6\Big).
\end{equation}	
We choose a fixed neighborhood structure, $\N_2$, given by {\em all} the two-body neighborhoods that are available in this six-qubit system. While our initial identification of this state and QL constraint was guided by numerical investigation in the context of the SDP approach we discuss in Sec. \ref{sec:sdp} \cite{code}, we here show in fully analytical fashion that $\qure_6$ is indeed UDA with respect to $\N_2$, yet it cannot occur as the UGS of any two-local Hamiltonian.

\subsubsection{$\qure_6$ is UDA relative to two-body neighborhoods}

We begin by characterizing the DQLS subspace of $\qure_6$. According to Eq.~\eqref{eq:DQLSstate}, we only need to consider the {\em support} of all the two-body RDMs for this purpose. Thanks to the symmetries that $\qure_6$ enjoys, it can be verified that its RDMs on NN neighborhoods coincide with each other. The same observation holds for non-NN RDMs as well. Let us examine the expectation values of the two-qubit projectors onto the standard basis $\{|00\rangle,|01\rangle,|10\rangle,|11\rangle \}$ of the two-qubit space $\H_2$, for  different combinations of qubit pairs, with respect to $\qure_6$. This lets us conclude that, for NN subsystems, 
\begin{align}
\label{eq:RDM_NN}
& \hspace*{2mm}\ker(\rho_{ij}) = \Span\{|11\rangle \} \;\Rightarrow  \nonumber \\ 
& \supp(\rho_{ij}) = \Span\{|00\rangle,|01\rangle,|10\rangle\}, \quad \forall j = i+1.
\end{align}
Similarly, in the non-NN case, 
\[  \ker(\rho_{ij}) = \varnothing  \Rightarrow \supp(\rho_{ij}) = \H_2, \quad \forall \,j > i+1. \]
Since the support of the non-NN RDMs coincides with the full space, such RDMs contribute trivially to the DQLS subpace $\H_{\N_2}(\qure_6)$, following Eq.~\eqref{eq:DQLSsubspace}. Therefore,
\begin{equation}
\H_{\N_2}(\qure_6) = \H_{\N_{\rm NN}}(\qure_6), 
\label{HNN}
\end{equation}
where $\N_{\rm NN}$ is the NN neighborhood structure under periodic boundary conditions. It follows that, in order to determine $\H_{\N_2}(\qure_6)$, we only need to characterize the DQLS subspace of $\qure_6$ relative to $\N_{\rm NN}$. 

To this end, consider the standard basis for the 64-dimensional six-qubit Hilbert space $\H=\H_6$, namely, 
\begin{equation*}
\mathcal{B}_6 = \{|0\rangle_6,|i_1\rangle_6,|i_1i_2,\rangle_6, \dots,|i_1i_2\dots i_6\rangle_6\},
\end{equation*}  
where each $i_k\in \{ 1, \ldots,6\}$ denotes the position of an excitation. We now determine whether each of these basis elements belongs to $\H_{\N_{\rm NN}}(\qure_6)$. Let $|\phi\rangle \in \H_6$, with the corresponding NN RDMs denoted by $\sigma_{i(i+1)}$. Then $|\phi\rangle \in \H_{\N_{\rm NN}}(\qure_6)$ only if 
\begin{eqnarray}
\label{eq:supp_ij}
\supp(\sigma_{i(i+1)}) & \subseteq & \supp(\rho_{i(i+1)})  \\ 
& = & \Span\{|00\rangle,|01\rangle,|10\rangle \},\quad 
\forall i =1,\dots,6, \nonumber 
\end{eqnarray}
with the second equality following from Eq.~\eqref{eq:RDM_NN}. Thus, it immediately follows that $|0\rangle_6 \in \H_{\N_{\rm NN}}(\qure_6)$, as its NN RDMs are equal to $|00\rangle\langle 00|$, thereby satisfying Eq. \eqref{eq:supp_ij}.  We also observe that for all one-excited basis states $\{|i_1\rangle_6\}$, the NN RDMs are either $|00\rangle\langle 00|$, $|10\rangle\langle 10|$ or $|01\rangle\langle 01|$, depending on the neighborhood chosen. Thus, 
\( |i_1\rangle_6 \in \H_{\N_{\rm NN}}(\qure_6).\)

A similar reasoning applies for any two-excited basis states with non-NN excitations. Therefore,
\( |i_1i_2\rangle_6 \in \H_{\N_{\rm NN}}(\qure_6)\), as long as \( \vert i_1-i_2\vert >1.\) Two-excited basis states of the form $|i(i+1)\rangle_6$ do {\em not} belong to the DQLS subspace, for the RDM $\sigma_{i(i+1)} = |11\rangle\langle 11|$, and therefore the support condition in Eq.~\eqref{eq:supp_ij} is not satisfied. This also includes the state $|16\rangle_5$, given that we are considering periodic boundary conditions. 

Following the same arguments, three-excited basis states $|135\rangle,|246\rangle \in \H_{\N_{NN}}(\qure_6)$. However, any other state from the basis set $\mathcal{B}$ does not belong to the DQLS subspace as their RDMs in the appropriate neighborhoods have non-rivial action on the subspace spanned by $|11\rangle$. In summary, we conclude that 
\begin{eqnarray}
\label{eq:H_NN}
&& \H_{\N_{\rm NN}}(\qure_6) = \Span\Big\{|0\rangle_6,|i\rangle_6,|i_1i_2\rangle_6,
|135\rangle_6,|246\rangle_6 \Big\} , \nonumber \\
&&\quad i,i_1,i_2=1,\dots,6;\, i_2\neq i_1+1;\, \{i_1,i_2\}\neq \{1,6\}.
\end{eqnarray}
Hereafter we shall use the notation $(i_2-i_1)>1$ to refer to the above allowed (nine) non-NN pairs of qubits.

In view of the equality in Eq. \eqref{HNN}, Eq.~\eqref{eq:H_NN} characterizes $\H_{\N_{2}}(\qure_6)$ as well.  We now focus on this 18-dimensional subspace $\H_{\N_{2}}(\qure_6) $ to determine the UDA nature of $\qure_6$. We do so by utilizing the concept of UDA witness introduced in Theorem~\ref{th:UDAwitness}: 
\begin{thm}
\label{thm:UDA}
$\qure_6$ is UDA relative to $\N_2$ as the QL operator 
\begin{equation}
\label{eq:M}
{\mathbb W}_6\equiv  \!\!\sum_{\substack{(i_2-i_1)>1}}{ \!\!(\sigma_{i_1}^+\sigma_{i_2}^+ + \sigma_{i_2}^-\sigma_{i_1}^- })
\end{equation}
serves as a $2$-local UDA witness for the state. 
\end{thm}
\noindent
Since the proof is lengthy, we defer it to Appendix \ref{proofs}.

\subsubsection{$\qure_6$ is not UGS of any two-body Hamiltonian}

Having established that $\qure_6$ is UDA relative to $\N_2$, we next look at the properties of $2$-local Hamiltonians for which $\qure_6$ is a ground state. We intend to show that such QL Hamiltonians have {\em at least a twofold degeneracy} for their ground-state space, thereby ruling out the possibility that $\qure_6$ may arise as their UGS.

Before doing so, we make an observation about the ground-state space of general (not necessarily QL) Hamiltonians, under the action of a group symmetrization operation \cite{LV99,PZ}: 
\begin{prop}
\label{pr:symm}
Let $H$ be a Hamiltonian with its ground-state space denoted by $g.s.(H)$. Assume that there exists a pure state $|\psi_g\rangle \in g.s.(H)$ and a finite group of unitary operations $\mathcal{G} \equiv \{G_1,\dots,G_{|{\cal G|}} \}$, such that $G_k|\psi_g\rangle = |\psi_g\rangle$, for all $G_k\in\mathcal{G}$. Let a new ${\cal G}$-symmetrized Hamiltonian be constructed from $H$ via 
\begin{equation}
\label{eq:barH}
\overline{H}  \equiv \frac{1}{|{\cal G}|} \sum_{k = 1}^{ |{\cal G}|} G^\dagger_k H G_k . 
\end{equation}
Then we have that 
\(g.s.(\overline{H}) \subseteq g.s.(H).\)
\end{prop} 

\begin{proof}
Assume, without loss of generality, that the ground-state energy of $H$ is zero. This implies $\overline{H}\geq0$, as each term in the sum of Eq.~\eqref{eq:barH} has the same spectrum as $H$, by unitarity of each $G_k$. Next we notice that
\begin{eqnarray*}
\tr{}{\overline{H}|\psi_g\rangle\langle\psi_g|} &=& \frac{1}{|{\cal G}|} \sum_k \tr{}{HG_k|\psi_g\rangle\langle\psi_g| G_k^\dagger} \\ 
&=& \tr{}{H|\psi_g\rangle\langle\psi_g|} = 0,
\end{eqnarray*}
since by assumption $|\psi_g\rangle$ is invariant under $\mathcal{G}$. The above relation also establishes that $|\psi_g\rangle \in g.s(\overline{H})$, with a corresponding ground-state energy of zero, since $\overline{H}\geq0$ as already seen. Therefore, in order to fully characterize  $g.s.(\overline{H})$, we need to figure out which other pure states have zero expectation value with respect to $\overline{H}$.
	
Consider an arbitrary state $|\phi\rangle \in \H$. Then
\[\tr{}{\overline{H}|\phi\rangle\langle\phi|} = \frac{1}{|{\cal G}|}\sum_k \tr{}{HG|\phi\rangle\langle\phi| G_k^\dagger} \geq 0, \]
where the equality holds if and only if 
\begin{equation}
\label{eq:g.s.}
|\phi_k\rangle \equiv G_k|\phi\rangle \in g.s.(H),\quad \forall\, k.
\end{equation}
A special instance of Eq.~\eqref{eq:g.s.} occurs for $G_k = I$, which then implies $|\phi\rangle\in g.s.(H)$. Thus 
\[g.s.(\overline{H}) \subseteq g.s.(H), \]
and the two subspaces coincide if and only if Eq.~\eqref{eq:g.s.} holds for all ground-state elements of $H$.
\end{proof}

By construction, the symmetrized Hamiltonian in Eq. \eqref{eq:g.s.} is projected onto the commutant of the group algebra ${\mathbb C}{\cal G}$, hence in particular it commutes with all the unitaries in $\mathcal{G}$ \cite{LV99,PZ}. Returning to our problem, we now examine how such a symmetrization procedure can be useful for our analysis. The state $\qure_6$ has a number of symmetry properties such as invariance under cyclic permutations  as well as reflections of qubits. If we are able to find a unitary group that leaves $\qure_6$ invariant, while preserving the QL nature of Hamiltonians relative to $\N_2$, we can thus restrict our investigation to the class of symmetrized QL Hamiltonians and exploit their simpler structure. Proposition~\ref{pr:symm} shows that such symmetrized Hamiltonians necessarily have $\qure_6$ in their ground-state spaces. Therefore, if we can show that there exists no symmetrized QL Hamiltonian for which $\qure_6$ is the UGS, we may infer that such a scenario is impossible for general QL Hamiltonians as well.

Consider the group of cyclic permutations on six-qubits, namely, \( \mathcal{P} \equiv  \{ P, P^2,\dots,P^6 = I\},\) with $P$ representing a unitary operator in $\H_6$ whose action is to permute the qubits in the order $(123456)\mapsto (234561)$. Similarly, we also consider the group generated by the (unitary) reflection operator $R$ whose operation is to swap the qubits as follows: $(123456) \mapsto (654321)$.  We denote this latter group by \(\mathcal{R} \equiv \{R,R^2 = I \}.\) Now, let
\begin{equation}
\label{eq:G}
\mathcal{G}_6 \equiv \langle R,P\rangle
\end{equation}
be the finite group generated by $P,R$ together, namely, the set of unitary operators on $\H_6$ that can be written 
as products of elements in ${\cal P}$ or ${\cal R}$. Clearly, 
\[G_k(\qure_6) = \qure_6,\quad \forall \,G_k \in \mathcal{G}_6. \] 
It is essential to note that all the group operations  in $\mathcal{G}$ are decomposable into a series of two-qubit swap operations, that are non-entangling when acted up on product basis states. Therefore, basis transformations by the unitaries belonging to $\mathcal{G}_6$ do not alter the QL structure of the operators relative to the neighborhood structure $\N_2$.

We now consider the class of QL Hamiltonians that have $\qure_6$ in their ground-state spaces {\em and} are symmetrized with respect to $\mathcal{G}_6$ given in Eq.~\eqref{eq:G}, to arrive at the following theorem: 
\begin{thm}
\label{thm:notUGS}
There exists no Hamiltonian $\overline{H}$ that is QL relative to $\N_2$ and symmetrized with respect to the 
group $\mathcal{G}_6$, such that $\overline{H}$ has $\qure_6$ as its UGS.
\end{thm}

\noindent
While a complete proof is deferred to Appendix \ref{proofs}, this establishes $\qure_6$ as a counter-example to the conjecture that UDA pure states are the UGS of a QL Hamiltonian.

\smallskip

\begin{remark}
This theorem also highlights the fact that for a general $\pure$ under a given neighborhood structure $\N$, having a UDA witness is a \emph{strictly} weaker property than  it being the UGS of a QL Hamiltonian.
In particular, notice that $\ket{\Psi}_6$ is \emph{not} an eigenstate of the witness $\mathbb{W}_6$ given in Eq.~\eqref{eq:M}. In fact, if $\pure$ is the not the UGS of any QL Hamiltonian, it has to be case that it is also not the eigenstate of its UDA witness $\mathbb{W}$, assuming the latter exists. To see this, observe that the frustration-free QL Hamiltonian constructed as a sum of the projectors onto $\textup{supp}(\rho_{\N_k})$, for all $\N_k\in\N$, has the DQLS subspace $\H_{\N}(\pure)$ as its ground state space~\cite{Ticozzi2012}. Now, $\mathbb{W}$ has higher energy for all states but $\pure$ in $\H_{\N}(\pure)$. Thus, an appropriate linear combination of these two QL operators will have $\pure$ as the UGS, if $\pure$ is an eigenstate of $\mathbb{W}$. The crucial difference to be noted here is that while being UGS forces the pure state to be the extremal eigenstate of a QL Hamiltonian, in order to be UDA it is sufficient that the state extremizes a QL Hamiltonian restricted to the DQLS subspace, without being its eigenstate. 
\end{remark}


\section{UGS vs. UDA properties: Semidefinite programming approach}
\label{sec:sdp}

Having established the inequivalence between the two properties of UDA and UGS by counterexample, we present a framework that allows exploring the relation between these two classes of pure states more generally, while also hinting at how such counterexamples may arise, at least from a mathematical standpoint. 

As we mentioned in the Introduction, the key idea is to describe the search problem for determining the UDA nature of a general pure state $\pure\in\H$ relative to a non-trivial $\N$ as a SDP. SDPs form a subclass of convex optimization problems \cite{Boyd} that are ubiquitous in engineering and physics. They emerge naturally in the quantum context, as the sets of density operators and quantum channels can be described by intersections of convex cones with linear constraints. The existence of a number of efficient solvers also makes SDP a very practical numerical optimization tool \cite{code}. After formalizing the UDA problem as an SDP, we construct the dual of the original problem and show that it can be interpreted as the optimization program to find a QL Hamiltonian for which $\pure$ has the lowest energy expectation value, subject to certain convex constraints. This dual program may or may not attain its optimal value depending on $\pure$. As a main result, we show that for a UDA pure state $\pure$ to be also a UGS of some QL Hamiltonian, relative to a same, fixed $\N$, it is {\em necessary and sufficient} that the dual optimal value is indeed attained. 

\begin{table*}[t]
	\normalsize
	\centering
	\begin{tabular}{l l l}
		(a) \quad \: & {\tt Minimize:} & $\tr{}{\dpure \sigma},$ \\
		                  & {\tt subject to:} &  (i.a) $\Phi_{\N}(\sigma) = \Phi_{\N}(\dpure)$, \\
			         &	 & (ii.a) $\sigma \geq 0.$
	\end{tabular} 
	\hspace*{1.1cm}
	\begin{tabular}{l l l}
		(b) \quad \: & {\tt Maximize:} &  $ -\tr{}{H\dpure}$, \\
		& {\tt subject to:} & (i.b) $H+\dpure\geq 0$, \\
		&	 & (ii.b) $H = \Phi_{\N}(H),$ \\
		& 	 & (iii.b) $H =H^\dagger.$
	\end{tabular} 
	\caption{SDP programs relevant to probe the relationship between UDA and UGS properties of 
	multipartite pure states. Left panel, (a): UDA primal problem. Right panel, (b): UDA dual problem. 
	See Appendix~\ref{sec:SDP} for additional detail on derivations.} 
	\label{tab:sdp}
\end{table*}

\subsection{The UGS problem as a dual optimization problem}

It is easy to see why the problem of determining the UDA nature of $\pure$ can be formulated as a SDP. If $\pure$ is non-UDA relative to $\N$, there exists another quantum state $\sigma$ such that $\sigma \in \mathcal{M}_{\N}(\pure)$. Clearly, the search space for such a $\sigma$, which is nothing but $\mathcal{D}(\H)$, is a convex set. More importantly, the requirement that the list of RDMs in $\N$ of $\sigma$ and $\dpure$ must coincide also imposes a  linear constraint on $\mathcal{B}(\H)$, as we explain next.

We first observe that the mapping from a quantum state $\rho \in \mathcal{D}(\H)$ to the list of its RDMs on $\N$, given by $\{\rho_{\N_k},\, \forall\N_k\in\N  \}$, can be described in terms of an orthogonal projector, say, $\Phi_{\N}: \mathcal{B}(\H) \to \mathcal{B}(\H)$. To see this, consider the operator basis for $\mathcal{B}(\H_a)$ given by 
\[\mathcal{X}_a \equiv \{I_a,X_{i_a}^a:i_a = 1,\dots,d_a^2-1 \},  \]
where $\{X_{i_a}^a\}$ are traceless Hermitian operators that are orthonormal with respect to the Hilbert-Schmidt norm, that is, up to the imaginary unit, they form a Lie algebra $\mathfrak{su}(d_a^2-1)$. Let then $\mathcal{X}$ denote the orthonormal basis for $\mathcal{B}(\H)$ obtained from $\{\mathcal{X}_a\}$ via the standard $N$-fold tensor-product construction \cite{Nielsen-Chuang}. For every fixed $X_i \in \mathcal{X}$ define a mapping $\Phi_{X_i}$ 
$: \mathcal{B}(\H) \to \mathcal{B}(\H) $ 
as follows:
\[\Phi_{X_i}(M) = \tr{}{X_iM}X_i,\quad \forall \, 
M\in\mathcal{B}(\H). \]
Evidently, $\Phi_{X_i}$ is a projector onto $\mathcal{B}(\H) $ since $\Phi_{X_i}^2 = \Phi_{X_i}$, following the orthonormality of the basis vectors in $\mathcal{X}$. 

For a fixed neighborhood $\N_k\in\N$, we can now define a basis set for the subspace  $\mathcal{B}(\H_{\N_k}) \otimes I_{\overline{\N}_k}$ by letting $\mathcal{X}_{\N_k} \equiv \{X_i\in\mathcal{X}: X_i = X^i_{\N_k}\otimes I_{\overline{\N}_k}\}$. The basis elements in $\mathcal{X}_{\N_k}$ can then be used to define the  RDM $\rho_{\N_k}$ via 
\[\rho_{\N_k}\otimes I_{\overline{\N}_k} \equiv  \sum_{X_j\in\mathcal{X}_{\N_k}} \Phi_{X_j}(\rho). \]
For the full neighborhood structure $\N = \{\N_k \}_{k=1}^M$, we can finally consider the union  \( \mathcal{X}_{\N} \equiv \cup_{\N_k\in\N}\mathcal{X}_{\N_k} \) and define a corresponding mapping $\Phi_{\N}$ such that 
\begin{equation}
\Phi_{\N} \equiv \sum_{X_i\in \mathcal{X}_{\N}} \Phi_{X_i}. 
\label{eq:qlp}
\end{equation}
It can be seen that $\Phi_{\N}$ is an orthogonal projector on $\mathcal{B}(\H)$ because it enjoys the following properties:
\begin{align*}
& \Phi_{X_i}\Phi_{X_j} = \delta_{ij}\Phi_{X_i},\: \forall \, X_i,X_j\in\mathcal{X}_{\N} \:\Rightarrow  \:\Phi_{\N}^2 = \Phi_{\N},\\ 
& \tr{}{\!M_1^\dagger\Phi_{\N}(M_2)\!} \!=\! \mathrm{Tr}\Big({[\Phi_{\N}(M_1)]^\dagger M_2}\Big), \quad\forall M_i\in \mathcal{B}(\H).
\end{align*}
By construction, $\Phi_{\N}$ is the desired mapping in $\mathcal{B}(\H)$ that effects the reduction from $\rho$ to its list of RDMs $\{\rho_{\N_k} \}$ on $\N$. It is immediate to verify the following:
\begin{coro}
Given any two states $\rho,\sigma\in\mathcal{D}(\H)$, the set of RDMs relative to any non-trivial $\N$ coincide, $\{\rho_{\N_k},\forall\,\N_k\in\N \} = \{\sigma_{\N_k},\forall\,\N_k\in\N \}$, if and only if their QL projections coincide, \( \Phi_{\N}(\sigma) = \Phi_{\N}(\rho).\)
\end{coro}

Following the above, the SDP for determining the UDA nature of $\pure$ relative to $\N$ is summarized in Table~\ref{tab:sdp}(a). The resulting primal program, and its dual, are in fact a pair of {\em linear} programs in standard form \cite{Boyd}. The aim of the primal program is to search for a quantum state $\sigma$ which is in $\mathcal{M}_{\N}(\pure)$, such that it has the minimum Hilbert-Schmidt inner product with $\dpure$. If such a  state exists and is not $\dpure$ itself, then clearly $\dpure$ is not UDA. Let the optimal value of the primal problem be denoted by $\alpha$. It is easy to see that  $\alpha \in[0,1 ]$ depending on the inputs $\pure$ and $\N$. Specifically, $\alpha = 1$ if and only if $\pure$ is UDA relative to $\N$, in other words $\mathcal{M}_{\N}(\pure) = \{\pure \}$. If $\alpha = 0$, it means that there exists some $\sigma \in \mathcal{M}_{\N}(\pure)$ such that the two states are orthogonal to each other. Notice that constraint (i.a) ensures that $\sigma$ has unit trace, since the identity is included in the basis set $\mathcal{X}_{\N}$ which we use to construct $\Phi_{\N}$ in Eq. \eqref{eq:qlp}. Combined with constraint (ii.a), this guarantee that the SDP is indeed searching over ${\cal D}(\H)$. Also note that the optimal value $\alpha$ is {\em always attained} irrespective of the problem inputs, as a consequence of the fact that the feasibility set of the primal problem, which is $\mathcal{M}_{\N}(\pure)$, is a non-empty, closed and convex set \cite{Boyd}.

Following the procedure outlined in Appendix~\ref{sec:SDP}, we can construct the dual to the UDA primal problem, resulting in the SDP given in Table~\ref{tab:sdp}(b). This dual problem searches for a Hermitian operator $H,$ which can be thought as a Hamiltonian, that has minimum energy expectation value with respect to $\pure$, subject to the constraints that $H$ is QL relative to $\N$ and $H+\dpure$ is a positive-semidefinite operator. Constraint (i.b), which also results from the dual problem construction, does not lend itself to an immediate physical interpretation. It will nonetheless be instrumental to showing that {\em if} a solution to the dual problem exists, the operator $H$ that optimizes the same is guaranteed to have $\pure$ as its UGS.

Let the optimal value for the dual program be given by $\beta$. Appendix~\ref{sec:SDP} shows that strong duality holds, namely, that $\alpha=\beta,$ for all $\pure$, $\N$, thanks to the refined version of Slater's condition for programs with affine inequality constraints \cite{Boyd}. Hence, it also follows that $\beta\in[0,1]$. In particular, we are interested in the case where $\pure$ is UDA relative to $\N$, that is, $\alpha = \beta = 1$. However, it is important to appreciate that strong duality for a primal-dual pair does {\em not} guarantee that the latter attains the optimum in general: the crucial difference between the primal and the dual problem is that the set over which we are optimizing the dual is {\em unbounded}. Thus, while the superior limit for the dual functional is guaranteed to be equal to the primal optimum, this in certain cases may be achieved only in the limit of unbounded operators $H$. The following result, which complements Proposition \ref{pr:ugsuda}, can be established: 
\begin{thm}
\label{thm:UDA_UGS}
Let $\pure \in \H$ be UDA relative to a non-trivial $\N$, and assume that the dual SDP attains its optimal value for the pair $\pure, \N$. Then there exists a Hamiltonian that is QL relative to $\N$ for which $\pure$ is the UGS.
\end{thm}
\begin{proof}
Since $\pure$ is UDA with respect to $\N$, the primal problem given in Table~\ref{tab:sdp}(a) has the optimal value $\alpha = 1$. The optimal solution is given by $\sigma = \dpure$ itself. Thanks to Slater's condition, the dual optimal value is also $\beta=-1$. Following our assumption, the latter is attained for some Hermitian operator that is QL relative to $\N$, say, $H_0$; accordingly, we have $ \tr{}{H_0\dpure} = -1$.

Assume that there exists $|\phi\rangle\langle\phi|\neq \dpure$ such that $|\phi\rangle \in g.s.(H_0)$, namely, $E_0\equiv\tr{}{H_0|\phi\rangle\langle\phi|} \leq -1$. Hence, 
\[ \tr{}{|\phi\rangle\langle\phi|(H_0+\dpure)} = \vert \langle \phi\pure\vert^2+E_0 <0, \]
as the overlap of the  two quantum states is strictly less than 1. This, however, contradicts the constraint (i.b) of the dual program given in Table~\ref{tab:sdp}(b) and hence the only ground state of $H_0$ must be $\pure$. 
\end{proof} 

Based on the above analysis, it follows that a UDA pure state is also the UGS of a QL Hamiltonian, provided the dual program attains its optimal value. However, attainment of the dual optimal value does not always happen, in which case the relevant UDA state can only belong to a degenerate ground space of a QL Hamiltonian. This is precisely what happens for the state $\qure_6$ we exhibited in Sec. \ref{sec:nonUDA}. 

On the contrary, we know [Proposition \ref{pr:ugsuda}] that pure states which are the UGS of QL Hamiltonians are always UDA relative to the same neighborhood structure. 
We can also see this using the SDP primal-dual relationship. If there exists a Hamiltonian $H_0$ that is QL relative to $\N$ with $\pure$ as its UGS, one can show that the dual program for $\pure$ is optimized by some $xH_0+y\mathbbm{1}$, with appropriate values for $x,y\in \mathbb{R}$ such that $\beta = 1$. This is possible because the energy gap between the ground and the first excited states of the Hamiltonian can be suitably scaled so that constraint (i.b) in Table~\ref{tab:sdp}(b) is satisfied for all quantum states with energy above the ground state $\pure$. In turn, this means that $\alpha = 1$ for the primal program, and hence $\pure$ is UDA with respect to $\N$, as expected.

\subsection{Generalized vs. QL linear constraints}

The result established in Theorem \ref{thm:UDA_UGS} naturally prompts the question of what features in the structure of the inputs and (or) the constraints may possibly make the assumption that the dual SDP attains its optimum value obeyed for any $\pure$, $\N$. As it is clear from the steps taken to construct the projector $\Phi_{\N}$ in Eq. \eqref{eq:qlp}, at the heart of our SDP primal problem is the fact that the QL constraint on $\mathcal{B}(\H)$ is {\em linear} in nature. Therefore, it is possible to write a similar SDP primal-dual pair for {\em any} general constraint as long as it acts linearly on $\mathcal{B}(\H)$, and arrive at similar conclusions to those of the QL case. For this reason, one may more generally ask whether (and when so) pure states which are UDA relative to some specified linear constraints are the UGS of some Hamiltonian which respects these same constraints. The corresponding optimization problems will look similar to the ones given in Table~\ref{tab:sdp}, with the QL projector $\Phi_{\N}$ replaced appropriately.  For instance, this generalized setting is applicable to scenarios where ``generalized locality'' constraints may be attributable to restrictions to a preferred subspace of observables (a preferred semi-simple Lie algebra in the simplest case) and ``generalized RDMs'' are constructed by restricting the linear functionals that represent global quantum states to determine only expectation values of observables in such a distinguished subspace \cite{Barnum}.

In Ref.~\cite{Chen2012}, a simple linear constraint in a two-qubit setting is considered, which forces linear operators to be of the form $M = \alpha H_1+\beta H_2$, where $\alpha,\beta \in \mathbb{C}$ and $H_1,H_2\in \mathcal{B}(\H)$ are fixed two-qubit operators. A pure state $\pure$ in this setting is UDA if there exists no other quantum state with the same set of expectation values as $\pure$ for $H_1,H_2$. It is easy to check that, for Hamiltonians of the form given above, if the ground state is unique it has to be UDA as well. However, the authors could show, by using geometric arguments, that there exist UDA pure states in this setting, which are not the UGS of any Hamiltonians of the appropriate kind. In SDP parlance, these cases would amount to an explicit violation of the assumption used in Theorem~\ref{thm:UDA_UGS}. However, prior to our analysis it was not clear whether such ``outlier'' UDA states could possibly be a peculiarity of the generalized setting in question, as opposed to the more structured one that pertains to QL constraints in multipartite systems, as we have focused on.

Mathematically, the QL setting is indeed special for a number of reasons. Notably, the mapping from a pure quantum state $\pure\langle\psi|$ to its set of RDMs is more complex than merely considering a set of expectation values; unlike in the latter case, each RDM is a positive-semidefinite operator not only on the reduced but also the full state space \cite{Cones}. Similarly, the subset of QL Hamiltonians have a richer structure in $\mathcal{B}(\H)$. The QL constraint depends on the underlying tensor product structure as well as the neighborhood structure chosen for the multipartite system, whereas linear constrains of the form considered in \cite{Chen2012} reflect only the total dimension and general geometric properties of $\mathcal{B}(\H)$. Notwithstanding this additional structure, our analysis shows that the question of whether UDA pure states are always UGS of appropriate Hamiltonians still has a negative answer in general, the reason being rooted in the failure of the dual SDP to achieve its maximum possible value.


\section{Conclusions and Outlook}
\label{sec:end}

We have elucidated aspects of the interplay between (quasi-)local and global properties in quantum states, in relation to developing better tools for characterizing sets of quantum states that are compatible with a given collection of quantum marginals and, more specifically, for better understanding the distinctive properties that pure states uniquely determined by such marginals enjoy.  Our main findings may be summarized as follows:

(i) For arbitrary $\rho$ and $\N$, respectively specifying the (pure or mixed) state of interest and the relevant QL constraint, the search for states consistent with the assigned neighborhood-RDMs can be restricted to a subset of the full space, specifically, to density operators with support on a subspace which is solely determined from knowledge of the {\em RDM supports} [Theorem \ref{th:DQLSsubspace}]. 
Such a DQLS subspace, which can be thought of as the minimal subspace that is asymptotically stabilizable using purely dissipative Markovian dynamics, may have a dimensionality significantly smaller than the full Hilbert space, possibly independent upon system size -- as we explicitly show for generalized W states. In the important case where the target state $\rho=|\psi\rangle\langle\psi|$ is pure, this may be further exploited to introduce the concept of a {\em UDA witness}, as an observable whose expectation is {\em uniquely extremized} by $\pure$ over the corresponding DQLS subspace [Theorem \ref{th:UDAwitness}]. As a by-product, generalized W states are shown to be UDA relative to arbitrary non-trivial QL constraints, extending and unifying existing results.  

(ii) While any non-degenerate ground state of a Hamiltonian that is QL relative to $\N$ is always UDA by its $\N$-RDMs, and while many UDA pure states are known to arise as the UGS of a corresponding QL parent Hamiltonian, we show this equivalence to be {\em false} in general -- disproving the conjecture that UDA alone suffices for $\pure$ to also be a UGS. We do so by exhibiting an analytic counter-example of a six-qubit pure state that is provably UDA by knowledge of its two-body RDMs [Theorems \ref{thm:UDA}] yet {\em cannot} be the non-degenerate UGS of any two-body Hamiltonian [\ref{thm:notUGS}].

(iii) We recast the problems of determining whether $\pure$ is UDA or, respectively, UGS relative to a give QL constraint as a {\em primal-dual pair} of linear programs in the general SDP framework. As a main result of this reformulation, we establish that a necessary and sufficient condition for a UDA pure state $\pure$ to be a UGS is for the dual problem to attain its optimal value [Theorem \ref{thm:UDA_UGS}]. Failure for this to happen, in spite of strong duality to be obeyed by the primal-dual pair, is identified as the mechanism precluding, in general, the equivalence between UDA and UGS to hold.

The present analysis leaves a number of interesting open questions for future investigation. In view of the fact that UGS pure states form a proper subset of UDA pure states, it is natural to ask what additional properties, on top of being UDA, a pure quantum state should obey in order for the equivalence to be possibly regained relative to a specified QL constraint. Or, even with UDA being the only property in place, it may be natural to ask instead what minimal ``coarse-grained'' neighborhood structure $\N'$ \cite{SaliniJPA} (if any) may allow for the UGS property to be guaranteed as well. From a SDP standpoint, elucidating these questions calls for a deeper understanding of the geometric shape that the optimization problem acquires as a result of the constraints -- in particular, the emergence and nature of non-exposed faces \cite{Chen2012,ChenJMP}. 

Finally, it is worth noting that, in many ways, UGSs of QL Hamiltonians ``mirror'' unique steady states of QL dissipative evolutions \cite{Zoller}; in particular, they may be prepared by using purely dissipative QL stabilizing dynamics if their parent Hamiltonian is frustration-free \cite{Ticozzi2012}. A related question one may thus ask is whether UGSs of QL Hamiltonians can always be stabilized with QL resources: interestingly, the answer turns out to be negative in this case also \cite{QLS}. 


\begin{acknowledgments}
It is a pleasure to thank Abhijeet Alase for useful discussions during this project. F.T. also wishes to thank Mattia Zorzi for valuable input about SDP. Work at Dartmouth was partially supported by the NSF under grant No. PHY-1620541, and the Constance and Walter Burke {\em Special Projects Fund in Quantum Information Science}.  F.T. has been partially funded by the QCOS project of the Universit\`a degli Studi di Padova.
\end{acknowledgments}

\begin{appendix}

\section{Technical proofs}
\label{proofs}
	
In this section, we present complete technical proofs of results stated in the main text. First, we provide an alternative proof of the UDA nature of arbitrary SLOCC equivalents to the W state [Theorem \ref{th:GW_UDA}], that exploits directly the structure of the DQLS subspace.

\medskip

\noindent 
{\bf Corollary III.5$'$}  
{\em SLOCC equivalent states of the $N$-qubit W state are UDA relative to any non-trivial neighborhood structure $\N$.}

\begin{proof}
Following Theorem~\ref{th:DQLSsubspace}, we know that
\[\supp(\sigma) \subseteq \H_{\N}(|{\rm GW}\rangle_N), \quad \forall \, \sigma \in \mathcal{M}_{\N}(|{\rm GW}\rangle_N). \]
The structure of $\H_{\N}(|{\rm GW}\rangle_N)$ is given by Eq.~\eqref{eq:GW_DQLS} in the main text. Hence we can express $\sigma$ in the following way: 
\[\sigma = p|\phi_0\rangle\langle\phi_0|+ (1-p)|\phi_1\rangle\langle\phi_1|, \quad 0\leq p\leq 1, \]
where $|\phi_0\rangle = \alpha |0\rangle_N +\beta |{\rm \overline{W}}\rangle_N$,  $|\phi_1\rangle = \beta^* |0\rangle_N -\alpha^* |{\rm \overline{W}}\rangle_N$ are a pair of orthonormal vectors with $|\alpha|^2+|\beta|^2=1$.
			
Let us now write the density operators $\rho \equiv |{\rm GW}\rangle_N\langle{\rm GW}|$ and $\sigma$ in terms of the states $|0\rangle_N, |{\rm \overline{W}}\rangle_N$:
\begin{eqnarray*}
\rho &= & c_0^2|0\rangle_N\langle 0|+
(1-c_0^2)|{\rm \overline{W}}\rangle_N\langle{\rm \overline{W}}| \\
&+& c_0\sqrt{1-c_0^2} \, \left(|0\rangle_N\langle{\rm \overline{W}}|+ |{\rm \overline{W}}\rangle_N\langle 0|\right) , \\
\sigma &= &\left( p|\alpha|^2+(1-p)|\beta|^2 \right) |0\rangle_N\langle 0|  \\
&+ & \left(p|\beta|^2+(1-p)|\alpha|^2\right) |{\rm \overline{W}}\rangle_N\langle{\rm \overline{W}}| \\
&+ & \left((2p-1) \alpha\beta^* |0\rangle_N\langle{\rm \overline{W}}|+ {\rm H.c.}\right) ,
\end{eqnarray*}
where H.c. denotes the Hermitian conjugate. 
Now, consider the RDMs of $\rho$ and $\sigma$ on some fixed $\N_j =\{k_1,\dots,k_L\}\in \N$. 
As these are identical by assumption, we denote the resulting operator by $\rho_{\N_j}$. Notice that:
(i) Terms of the form  $|k\rangle_L \langle k'|$ + H.c. arise only from the partial tracing of 
$|{\rm \overline{W}}\rangle_N\langle{\rm \overline{W}}|$.
(ii) Terms of the form  $(|k\rangle_L \langle 0|$ +H.c.)
arise only from the partial tracing $(|{\rm \overline{W}}\rangle_N\langle{0}|$+H.c.)
Hence the coefficients of these two terms must be equal in $\rho$ and $\sigma$. That is, it must be	
\begin{eqnarray}
p|\beta|^2+(1-p)|\alpha|^2 &=& 1-c_0^2,  
\label{GW1} \\
(2p-1) \alpha\beta^* & = &c_0\sqrt{1-c_0^2}, 
\label{GW2} 
\end{eqnarray}
where $\sqrt{1-c_0^2} > 0$, since it is the normalization constant used in the definition  of $|\overline{\rm W}\rangle_N$ (see  Eq.~\eqref{eq:GW} in the main text). Therefore, the right hand-side of Eq.~\eqref{GW2} is real and positive. This implies that $\alpha\beta^*$ is real, and the complex numbers $\alpha,\beta$ have the same phase factor which can be absorbed into the global phase of $|\phi_0\rangle,|\phi_1\rangle$. Without loss of generality, we then take $\alpha,\beta \in \mathbb{R}$ and $\alpha \geq 0$. 
		
From Eq.~\eqref{GW1}, it follows that 
\(c_0^2 =p|\alpha|^2+(1-p)|\beta|^2.  \)
Squaring Eq.~\eqref{GW2} and substituting for $c_0$ from this expression yields 
\[\alpha^2\beta^2(2p-1)^2 = (p^2+(1-p)^2)\alpha^2\beta^2+p(1-p)(\alpha^4+\beta^4), \]
which may be further simplified to:
\[2\alpha^2\beta^2p(p-1) = p(1-p)(\alpha^4+\beta^4). \]
For $p\not\in \{0,1\},$ we see that this yields $(\alpha^2+\beta^2)^2 =0,$ which is a contradiction. Therefore, it follows that either 
i) $p=0$, in which case $\alpha = \sqrt{1-c_0^2}$ and $\beta = -c_0$; or \\
ii) $p=1$, in which case $\alpha = c_0$ and $\beta = \sqrt{1-c_0^2}$. \\
Both conditions are equivalent to $\sigma =  |{\rm GW}\rangle_N\langle{\rm GW}|$.
This shows that $ |{\rm GW}\rangle_N$, and therefore the SLOCC class of the $N$-qubit W state, are UDA 
relative to any non-trivial $\N$, as claimed.
\end{proof}
	
\medskip

\noindent 
{\bf Theorem IV.2.} {\em 
$\qure_6$ is UDA relative to $\N_2$ as the QL operator 
\begin{equation}
{\mathbb W}_6\equiv  \!\!\sum_{\substack{(i_2-i_1)>1}}{ \!\!(\sigma_{i_1}^+\sigma_{i_2}^+ + \sigma_{i_2}^-\sigma_{i_1}^- })
\end{equation}
serves as a $2$-local UDA witness for the state. }

\begin{proof}
We aim to show that $\langle \Psi|{\mathbb W}_6 \qure_6 > \langle\phi| {\mathbb W}_6|\phi\rangle$, for all $|\phi\rangle\in \H_{\N_2}(\qure_6)$, where $ |\phi\rangle\langle\phi|\neq\qure_6\langle\Psi|$. It is clear that ${\mathbb W}_6$ is Hermitian and QL relative to $\N_2$.
	
Let a general normalized pure state $|\phi\rangle\in\H_{\N_2}(\pure)$ be expressed as follows (up to a global phase factor):
\begin{eqnarray}
\label{eq:phi}
|\phi\rangle & = &a_0|0\rangle_6+ (a_{135}+ib_{135})|135\rangle_6+(a_{246}+ib_{246})|246\rangle_6 \nonumber \\ 
&+& \sum_{i_1}(a_{i_1}+ib_{i_1})|i_1\rangle_6  
+  \hspace{-1ex}\sum_{(i_2-i_1) >1}\hspace{-1ex}(a_{i_1i_2}+ib_{i_1i_2})|i_1i_2\rangle_6, \nonumber  
\end{eqnarray}
where all the expansion coefficients are chosen to be real and, due to normalization, they obey
\(\sum_{i_1}(a_{i_1}^2+b_{i_1}^2)+
\sum_{i_2-i_1 >1} (a_{i_1i_2}^2+b_{i_1i_2}^2)+ a_{135}^2+b_{135}^2+a_{246}^2+b_{246}^2 = 1.\)
We now show that  $\qure_6$ is the unique state with maximum expectation value for ${\mathbb W}_6$, among all the pure 
states in $\H_{\N_2}(\qure_6)$. We use the method of Lagrangian multipliers. Define vectors 
\begin{align}
\vec{a} &\!\equiv\! \{a_0, a_{i_1}, a_{i_1i_2}, a_{135},a_{246}: i_\ell = 1,\dots,6,\, i_2-i_1>1\},\!
\label{eq:a} \\
\vec{b} &\!\equiv \!\{b_{i_1}, b_{i_1i_2}, b_{135},b_{246}: i_\ell = 1,\dots,6,\, i_2-i_1>1\}, \!
\label{eq:b}
\end{align}  
where $\ell =1, 2$, and two functions $f$ and $h$ as follows: 
\begin{eqnarray}
f(\vec{a},\vec{b}) &\equiv & 2\Big( a_0\!\! \!\sum_{(i_2-i_1)>1} \!\!a_{i_1i_2} \label{eq:f}\\
 &+& a_{135}(a_1+a_3+a_5) + a_{246}(a_2+a_4+a_6) \nonumber \\
 &+& b_{135}(b_1+b_3+b_5) +b_{246}(b_2+b_4+b_6)\Big), \nonumber \\
h(\vec{a},\vec{b}) &\equiv &\vert \vec{a}\vert^2+\vert \vec{b}\vert^2. \nonumber
\end{eqnarray}
Note that $f(\vec{a},\vec{b}) = \langle \phi| {\mathbb W}_6|\phi\rangle$ for all $|\phi\rangle \in \H_{\N_2}(\qure_6)$ following Eqs.~\eqref{eq:phi}-\eqref{eq:f}. Thus, we can also use $f(|\phi\rangle)$ to represent the expectation of ${\mathbb W}_2$ with respect to $|\phi\rangle$.
	
Our aim is to maximize $f(\vec{a},\vec{b})$ subject to the constraint that $h(\vec{a},\vec{b}) = 1$, which represents the normalization condition on $|\phi\rangle$ in Eq.~\eqref{eq:phi}. Consider the Lagrangian
\[
\L(\vec{a},\vec{b},\lambda) = f(\vec{a},\vec{b})+\lambda \,h(\vec{a},\vec{b}), \quad \lambda\in\mathbb{R}. \] 
We wish to solve the equations given by
\(\nabla_{\vec{a},\vec{b},\lambda}\L = 0,\)
where $\nabla_{\vec{a},\vec{b},\lambda}$ represents the partial derivatives of $\L(\vec{a},\vec{b},\lambda)$ with respect to each real-valued component of the vectors $\vec{a},\vec{b}$, given by Eqs.~\eqref{eq:a}-\eqref{eq:b}, as well as the scalar variable $\lambda$. Explicitly, the resulting set of equations reads as follows:
\begin{eqnarray}
	&&\!\!\! 2 \lambda a_0 + 2\sum_{(i_1-i_2 ) >1} a_{i_1i_2} = 0 ,
	\label{eq:1}\\
	&&\!\!\! 2a_0+2\lambda a_{i_1i_2} = 0,\quad \forall \,(i_1-i_2) >1 ,
	\label{eq:2} \\
	&&\!\!\! 2\lambda b_{i_1i_2} = 0, \quad \forall \,(i_1-i_2) >1 ,
	\label{eq:3} \\
	&&\!\!\!2(a_{i_1}+a_{i_2}+a_{i_3})+2\lambda a_{i_1i_2i_3} = 0, \nonumber \\ 
	&&\; \{i_1,i_2,i_3\} = \{1,3,5\},\{2,4,6\} ,
	\label{eq:4} \\
	&&\!\!\!2a_{i_1i_2i_3} + 2\lambda a_{i} = 0, \nonumber \\ 
	&&\;\{i_1,i_2,i_3\} = \{1,3,5\},\{2,4,6\},\, i = i_1,i_2,i_3, 
	\label{eq:5} \\
	&&\!\!\!2(b_{i_1}+b_{i_2}+b_{i_3})+2\lambda b_{i_1i_2i_3} = 0, \nonumber \\ 
	&&\; \{i_1,i_2,i_3\} = \{1,3,5\},\{2,4,6\}, 
	\label{eq:6} \\
	&&\!\!\!2b_{i_1i_2i_3} + 2\lambda b_{i} = 0, \nonumber  \\ 
	&&\; \{i_1,i_2,i_3\}= \{1,3,5\},\{2,4,6\},\,i = i_1,i_2,i_3.
	\label{eq:7}
\end{eqnarray}
We first solve Eqs.~\eqref{eq:1}-\eqref{eq:2} in order to find the feasible values for $a_0,\lambda$. 
This leads us to the following cases:

\smallskip

$\bullet$ {\bf Case 1:} $a_0 \neq 0, \lambda = \pm 3$.	
Plugging in the relevant values, we get $a_{i_1i_2} = \pm a_0/3$ in Eq. \eqref{eq:2} and $b_{i_1i_2}=0$ in Eq.~\eqref{eq:3}, for all $(i_1-i_2) >1$. One may also verify that for the particular choice of $\lambda$, Eqs.~\eqref{eq:4}-\eqref{eq:5} and Eqs.~\eqref{eq:6}-\eqref{eq:7} are only satisfied if we set $a_{i_1i_2i_3},b_{i_1i_2i_3} = 0$ for all suitable $\{i_1,i_2,i_3\}$. As a consequence, the corresponding $a_i,b_i = 0$ in these equations.
		
The  solutions corresponding to this case are (up to an overall phase) given by:
\[|\phi_1^{\pm}\rangle =  a_0(|0\rangle_6 \pm |\overline{\text{D}}\rangle_6),\]
where $|\overline{\text{D}}\rangle_6$ is defined in Eq.~\eqref{eq:barD} of the main text. Upon normalizing we get 
$a_0 = 1/\sqrt{2}$. The corresponding functional value can be verified to be \(f(|\phi_1^{\pm}\rangle)= \pm 3.\) Since we are looking for the maximum of $f$, we only retain the solution $|\phi_1^{+}\rangle$, which is nothing but $\qure_6$ itself.
		
\smallskip

$\bullet$ {\bf Case 2:} $\lambda=\pm\sqrt{3}, a_0 = 0$.
Solving for this case leads to $b_{i_1i_2} = 0$ in Eq.~\eqref{eq:3}, for all $(i_2-i_1)>1$. Next, focus on Eqs.~\eqref{eq:4}-\eqref{eq:5}. Solving them yields $a_i =  a_{i_1i_2i_3}/\lambda$ for appropriate values of $i,\{i_1,i_2,i_3\}$. Similar relations hold for $b_{i_1i_2i_3}, b_i$, following Eqs.~\eqref{eq:6}-\eqref{eq:7}. Accordingly, there are two different solutions emerging from this case, which correspond to the two different $\lambda$ values. Their common  form is given below (again up to an overall phase):
\begin{eqnarray*}
|\phi_2^{\pm}\rangle &=& (a_{135}+ib_{135})(|135\rangle_6 \pm \frac{1}{\sqrt{3}}(|1\rangle_6+|3\rangle_6+|5\rangle_6)\\
&+& (a_{246}+ib_{246})(|246\rangle_6 \pm \frac{1}{\sqrt{3}}(|2\rangle_6+|4\rangle_6+|6\rangle_6), 
\end{eqnarray*}
together with the appropriate normalization condition, 
\(2(a_{135}^2+b_{135}^2+a_{246}^2+b_{246}^2) =1. \)
Computing the functional values in this case gives, 
\(f(|\phi_2^{\pm}\rangle) = \pm \sqrt{3}. \) Since both these values are lower than $f(|\phi_1^{+}\rangle)$ from 
Case 1, we discard the solutions obtained from this case.

\smallskip

$\bullet$ {\bf Case 3:} $\lambda = 0$. It follows that $a_0 = 0$ and \(\sum_{(i_2-i_1) >1}a_{i_1 i_2} = 0,  \)	after solving Eqs.~\eqref{eq:1}-\eqref{eq:2}. We are free to choose arbitrary values for $a_{i_1i_2}$, as long as the above relation is satisfied. Similarly, the variables given by $b_{i_1i_2}$ are also chosen freely, with no constraints, following Eq.\eqref{eq:3}. Solving Eqs.~\eqref{eq:4}-\eqref{eq:5} sets $a_{i_1i_2i_3} = 0$, with \(a_{i_1}+a_{i_2}+a_{i_3} = 0. \) Similar relations hold for $b_{i_1i_2i_3}, b_i$ in Eqs.~\eqref{eq:6}-\eqref{eq:7}. However, one can easily verify that the functional value in Eq.~\eqref{eq:f} is always 0 irrespective of the choices available for the non-zero variables. Therefore, we discard this case as well.

\smallskip

In summary, we conclude that the unique state which maximizes the expectation value of ${\mathbb W}_6$ is $\qure_6$, following the analysis given in Case 1. Hence ${\mathbb W}_6$ is a UDA witness for $\qure_6$ for the neighborhood structure $\N_2$ by Theorem.~\ref{th:UDAwitness}. 
\end{proof}

\medskip

\noindent 
{\bf Theorem IV.4.}
{\em There exists no Hamiltonian $\overline{H}$ that is QL relative to $\N_2$ and symmetrized with respect to the group $\mathcal{G}_6$, such that $\overline{H}$ has $\qure_6$ as its UGS.}

\begin{proof}
Since $\overline{H}$ is QL relative to $\N_2$, it can be written as 
\[\overline{H} = \sum_{\N_k\in\N_2} \overline{H}_{\N_k} \otimes I_{\overline{\N}_k },\]
where each $\overline{H}_{\N_k}$ is a two-qubit Hamiltonian. Since we are interested in the action of $\overline{H}$ on $\qure_6$, which is a superposition of states with specified (zero and two) excitations, it is convenient to represent $\overline{H}$ in terms of a product operator basis that makes the creation or annihilation of excitations explicit. Thus, we represent $\overline{H}$ in the following form:
\begin{equation}
\label{eq:barH_decomp}
\overline{H} = \overline{H}_2+\overline{H}_{-2}+\overline{H}_1+\overline{H}_{-1}+\overline{H}_0.	
\end{equation}
Here, $\overline{H}_2$ is a QL operator that creates two excitations and is composed of operators of the form 
\(\sigma^+_i\sigma^+_j \equiv |11\rangle \langle 00|,\, i,j\in\{1,\dots,N\}. \) The term that destroys two excitations is $\overline{H}_{-2} = \overline{H}_2^\dagger$. Similarly, we have $\overline{H}_1 ( \overline{H}_{-1}$) for creating (destroying) one excitation, with $\overline{H}_1 = \overline{H}_{-1}^\dagger$. $\overline{H}_1$ is composed of two-qubit operators of the form $\sigma^z_i\sigma^+_j$ or $\sigma^+_i\sigma^z_j$, or one-qubit operators $\sigma^+_i$. $\overline{H}_0$ is the excitation-preserving term, which consists of $\sigma^z$ operators in one or two qubits, operators of the form $\sigma^+_i\sigma^-_j$ or $\sigma^-_i\sigma^+_j$, or simply the identity operator. Clearly, $\overline{H}_0$ is Hermitian.
	
It is evident from the above analysis that the different terms in Eq.~\eqref{eq:barH_decomp} are linearly independent of each other. Therefore, they are all individually symmetric relative to the group $\mathcal{G}$ as their sum is. As these terms must then be invariant under cyclic permutations as well as reflection of qubits, they belong to operators subspaces spanned by the following basis sets:
\begin{eqnarray}
\hspace*{-7mm}&& \overline{H}_2 \!\in \Span_{\mathbb{R}}\{ P(\sigma^+_{1}\sigma^+_2\otimes I_{3456}),P(\sigma^+_{1}\sigma^+_3\otimes I_{2456}), 
 \label{eq:span_2exci} \\ 
\hspace*{-6mm}&& \hspace*{19mm}P(\sigma^+_{1}\sigma^+_4\otimes I_{2356})\!\},  \nonumber \\
\hspace*{-6mm}&& \overline{H}_1 \!\in \Span_{\mathbb{R}}\{P(\sigma^+_{1}\otimes I_{23456}),RP(\sigma^+_{1}\sigma^z_2\otimes I_{3456}), \label{eq:span_1exci} \\ 
\hspace*{-6mm}&& \hspace*{16mm}RP(\sigma^+_{1}\sigma^z_3\otimes I_{2456}), RP(\sigma^+_{1}\sigma^z_4\otimes I_{2356})\!\}, 
\nonumber  \\
\hspace*{-6mm}&& \overline{H}_0 \!\in \Span_{\mathbb{R}}\{ P(\sigma^z_{1}\sigma^z_2\otimes I_{3456}),P(\sigma^z_{1}\sigma^z_3\otimes I_{2456}),
\nonumber \\ 
\hspace*{-6mm}&& \hspace*{19mm}P(\sigma^z_{1}\sigma^z_4\otimes I_{2356}), P(\sigma^z_1\otimes I_{23456}),  \nonumber \\ 
\hspace*{-6mm}&& \hspace*{16mm}RP (\sigma_1^+\sigma_2^-\otimes I_{3456}), RP (\sigma_1^+\sigma_3^-\otimes I_{2456}), \nonumber \\
\hspace*{-6mm}&& \hspace*{16mm}RP (\sigma_1^+\sigma_4^-\otimes I_{2356}) \!\}. \nonumber 
\end{eqnarray}	
Without loss of generality we can assume that $\overline{H}\geq 0$ and, therefore, $\overline{H}\qure_6 = 0$ as it belongs to the ground-state space. Since $\qure_6 = 1/\sqrt{2}(|0\rangle_6+|\overline{\text{D}}\rangle_6)$ (see Eq.~\eqref{eq:barD} in the main text for the form of $|\overline{\text{D}}\rangle_6$), we infer that 
\begin{align}
\overline{H}_{2}|\overline{\text{D}}\rangle_6 &= 0 , \label{eq:2exci} \\
\overline{H}_{1}|\overline{\text{D}}\rangle_6 &= 0 ,  \label{eq:1exci} \\
\overline{H}_2|0\rangle_6 &= -\overline{H}_{0}|\overline{\text{D}}\rangle_6 , \label{eq:2exci0}\\
\overline{H}_1|0\rangle_6 &= -\overline{H}_{-1}|\overline{\text{D}} \rangle_6 ,\label{1exci}\\
\overline{H}_0|0\rangle_6 &= -\overline{H}_{-2}|\overline{\text{D}}\rangle_6 . \label{eq:0exci0}
\end{align}
Note that $\overline{H}_{-2}|0\rangle_6 = 0 = \overline{H}_{-1}|0\rangle_6 $ are trivially obeyed.

We first focus on solving Eq.~\eqref{eq:2exci}. Consider an eigenbasis of the unitary operator $P$, given by $\{|\phi_k\rangle \}$. Being $P$ a permutation operator, it does not create or destroy excitations when acting on any quantum state. For this reason and by exploiting the degeneracy in the spectrum of $P$, we may choose the eigenbasis of $P$ in such a way that each $|\phi_k\rangle$ is a linear combination of terms with a well-defined number of excitations, ranging between zero and six. We thus must have
\begin{equation}
\label{eq:phi_k}
\langle \phi_k|\overline{H}_{2}|\overline{D}\rangle_6 = 0,\quad \forall k. 
\end{equation}
Recalling that $|\overline{D}\rangle_6$ is a two-excitation state, Eq.~\eqref{eq:phi_k} is trivially satisfied unless $\{|\phi_k\rangle \}$  consists of four-excitation terms. Thus, we focus on $|\phi_k\rangle$s that are solely composed of four-excitation terms and examine what restriction they impose on the structure of $\overline{H}_2$. Notice that
\begin{equation*}
\langle \phi_k |\overline{H}_{2}|\overline{D}\rangle_6 = \langle \phi_k |P^\dagger\overline{H}_{2}P|\overline{D}\rangle_6 = \lambda_k\langle \phi_k |\overline{H}_{2}|\overline{D}\rangle_6 . 
\end{equation*}
The first equality holds because $\overline{H}_{2}$ is invariant under the action of $\mathcal{G}$. The second equality follows from the fact that $|\overline{D}\rangle_6$ is invariant under the action of $P$ and $|\phi_k\rangle$ is the eigenstate of $P$ with eigenvalue $\lambda_k$. This shows that when $\lambda_k \neq 1$, Eq.~\eqref{eq:phi_k} is automatically be satisfied without imposing any additional constraint on $\overline{H}_{2}$. Therefore, we further restrict our attention to the following states for which the eigenvalue $\lambda_k = 1$:
\begin{equation}
\label{eq:4exci}
|\phi_k\rangle \in \{P(|3456\rangle_6) ,P(|2456\rangle_6),P(|2356\rangle_6) \}. 
\end{equation}
Following Eq.~\eqref{eq:span_1exci}, let
\begin{eqnarray*}
\overline{H}_2 &\equiv & a_1P(\sigma^+_{1}\sigma^+_2\otimes I_{3456})+a_2 P(\sigma^+_{1}\sigma^+_3\otimes I_{2456}) \\ 
&+ &a_3 P(\sigma^+_{1}\sigma^+_4\otimes I_{2356}),
\end{eqnarray*} 
where $a_1,a_2,a_3\in\mathbb{R}$ are treated as unknowns. We then obtain three equations of the form $\langle \phi_k|\overline{H}_{2}|\overline{D}\rangle_6 = 0$, corresponding to the three states in Eq.~\eqref{eq:4exci}. This set of equations can be rewritten in a matrix form $A\vec{a}=\textbf{0}$, where each entry of the matrix is given by $A_{jk} = \langle \phi_k|P_j|\overline{D}\rangle$, with $P_i$ belonging to the set in Eq.~\eqref{eq:span_2exci} and $\vec{a} = (a_1,a_2,a_3)^T$. By evaluating the matrix $A$ with Matlab, we found that $\det A \neq 0$ and these equations are simultaneously satisfied only for $a_1=a_2=a_3=0$. Accordingly, $\overline{H}_2 = 0$, and similarly for $\overline{H}_{-2}$.
	
Next we carry out a similar analysis for Eq.~\eqref{eq:1exci}. Based on Eq.~\eqref{eq:span_1exci}, we parametrize $\overline{H}_1$ as follows:
\begin{eqnarray*}
\overline{H}_1 & \equiv & b_1P(\sigma^+_{1}\otimes I_{23456})+b_2RP(\sigma^+_{1}\sigma^z_2\otimes I_{3456}) \\ &+& b_3RP(\sigma^+_{1}\sigma^z_3\otimes I_{2456})+b_4 RP(\sigma^+_{1}\sigma^z_4\otimes I_{2356}),
\end{eqnarray*} 
with $b_1,b_2,b_3,b_4\in\mathbb{R}$ treated as unknowns. We observe that $\langle \phi_k|\overline{H}_{1}|\overline{D}\rangle_6 = 0$, similar to Eq.~\eqref{eq:2exci}, for $|\phi_k\rangle$ in the eigenbasis of $P$. However, this time we only consider
\begin{equation}
\label{eq:3exci}
|\phi_k\rangle \in \{P(|123\rangle_6) ,P(|134\rangle_6),P(|135\rangle_6) \}, 
\end{equation}
since the most general form of $\overline{H}_{1}|\overline{D}\rangle_6 $ can only have three excitations present. We can then form a matrix $B$ with elements $B_{jk} = \langle \phi_k|P_j|\overline{D}\rangle$, for $P_j$ belonging to the set in Eq.~\eqref{eq:span_1exci}. The set of equations is rewritten as $B\vec{b} = \textbf{0}$ with $\vec{b} = (b_1,b_2,b_3,b_4)^T$. In this case, one can verify that there exists one non-trivial solution for $\vec{b}=(-1,1,1,1)^T$, such that $\langle \phi_k|\overline{H}_{1}|\overline{D}\rangle_6 = 0$ holds for the choice of $|\phi_k\rangle$ in Eq.~\eqref{eq:3exci}. However, one may also verify that this solution fails to satisfy Eq.~\eqref{1exci} and therefore is not a valid choice for $\overline{H}_1$. Accordingly, we are left with $\overline{H}_1={0}$, and similarly for $\overline{H}_{-1}$.  

Based on the above analysis, we conclude that the relevant QL parent Hamiltonian $\overline{H}$ is excitation-preserving,  $\overline{H}= \overline{H}_0$. Then the only non-trivial equations that remain are Eqs.~\eqref{eq:2exci0},~\eqref{eq:0exci0}, which in turn imply that 
\begin{equation*}
\overline{H}|\overline{D}\rangle_6=0 = \overline{H}|0 \rangle_6.
\end{equation*}
This shows that the ground-state space of $\overline{H}$ is at least two-dimensional. Thus, $\qure_6$ is not UGS of any $\overline{H}$ that is QL relative to $\N_2$ and invariant under $\mathcal{G}$. 
\end{proof}

\section{Derivation of the dual problem}
\label{sec:SDP}

In many cases, in order to solve an SDP problem, it is convenient to derive its dual counterpart. In essence, this amounts to write a parametric lower bound on the primal problem, and maximize such lower bound on the set of parameters. The interest in this accessory, in some sense ``relaxed'' optimization problem may generally stem from two reasons: First, the new objective function is concave even if the original cost function is not convex, making the dual problem more tractable; second, under some conditions on the constraints, it can be shown that the optimal value for the dual functional corresponds to the optimal primal cost. Here, we explicitly construct the dual problem of our SDP check for UDA pure states, following the general approach presented in \cite{Boyd}.

In the primal problem, where the objective is to determine whether a pure state $\pure$ is UDA, we aim to minimize the functional $f(\sigma,\rho) = \text{tr}(\rho\sigma),$ with $\rho=\dpure,$ over the set of Hermitian positive-semidefinite matrices subject to a set of linear constraints subsumed in the linear map $\Phi_{\cal N}$:

\begin{center}
\begin{tabular}{l l}
	{\tt Minimize:} & $\text{tr}[\rho\sigma]$ \\
	{\tt subject to:} 	& $\Phi_{\cal N}(\sigma)=\Phi_{\cal N}(\rho),$ \\
	& $\sigma \geq 0.$
\end{tabular}
\end{center}

\noindent 
Recall that the minimum always exists since we are optimizing a linear function over a convex, non-empty bounded set ${\cal A}$, which is the admissible set for our problem. Let us call its value  $\alpha \equiv \min_{\sigma\in{\cal A}}f(\sigma, \rho)$. 

In order to find the lower bounds, that is, functionals of the dual problem, we first construct the Langrangian of the primal problem, by essentially incorporating the constraints as penalties in the  primal cost function:
$$ {\cal L} (\sigma,\lambda,\nu) \equiv  \text{tr}\Big[\rho\sigma-\lambda\sigma+H \Big(\Phi_{\cal N}(\sigma)-\Phi_{\cal N}(\rho) \Big) \Big], $$ 
with $\lambda \geq 0$ and $H$ Hermitian. The Lagrangian $ {\cal L}$ constructed in this way is such that for a given $\sigma  \geq 0$, its value is always lesser than or equal to that of $f$ by appropriately choosing the dual variable $H$. In order for the comparison to make sense, $H$ is consistently restricted to be Hermitian. The Lagrange dual function for $f(\sigma, \rho)$ is then given by 
$$g(\lambda,H) \equiv \inf_{\sigma \geq 0}  {\cal L}(\sigma,\lambda,H). $$ 
If $\tilde{\sigma}$ denotes a feasible point for the function $f(\sigma, \rho$) , that is, $\tilde{\sigma} \in \mathcal{A}$, it can be shown that 
$$f(\tilde{\sigma}) \geq  {\cal L} (\tilde{\sigma},\lambda,\nu) \geq 
\inf_{\sigma \in \mathcal{D}(\mathcal{H})}  {\cal L}(\sigma,\lambda,\nu) = g(\lambda,\nu) .$$ 
The first inequality follows by looking at the definition of the Lagrangian. It is the sum of $f(\sigma)$ and two other terms, of which $\text{tr}(-\lambda\sigma)$ is always negative for any feasible $\tilde{\sigma}$, since $\lambda \geq 0$ by choice. The second term, $\text{tr}[\nu(\Phi_{\cal N}(\sigma)-\Phi_{\cal N}(\rho))]$ is zero for any allowed $\tilde{\sigma}$ because of the equality constraint in the primal problem, and hence it does not contribute to $ {\cal L}$. Thus, $f(\tilde{\sigma}) \geq  {\cal L}(\tilde{\sigma},\lambda,H)$ for $\tilde{\sigma} \in \mathcal{A}$ and $\lambda \geq 0$. The second inequality is obvious. 

The dual problem is obtained by looking at the {\em best} lower bound that we can derive for the primal optimum using the Lagrange dual function. Let us denote by $\Phi^\dag_{\cal N}$ the dual of $\Phi_{\cal N}$. Notice that, in our case, one can in fact see that $\Phi^\dag_{\cal N}=\Phi_{\cal N}$, as both these maps orthogonally project their arguments to the corresponding QL coordinates. First, we rewrite the terms composing the function $g,$ so that its structure is more explicit:
\begin{align*}
g(\lambda,H) &= \text{tr}[ -H\Phi_{\cal N}(\rho))] +\inf_{\sigma \geq 0}\text{tr} [(\rho-\lambda)\sigma+H\Phi_{\cal N}(\sigma)]  \\
& = \text{tr}[-\Phi^\dag_{\cal N}(H)\rho]+\inf_{\sigma \geq 0}\text{tr} [(-\rho+\lambda+\Phi^\dag_{\cal N}(H))\sigma].     
\end{align*}
Accordingly, the Lagrange dual problem for our  $g(\lambda,H)$ can be written as:

\begin{center}
\begin{tabular}{l l}
	{\tt Maximize:} & $\text{tr}[\Phi^\dag_{\cal N}(-H)\rho]$ \\
	{\tt subject to:} & $H =H^\dagger$, \\
	& $\lambda \geq 0$
\end{tabular}
\end{center}

\noindent 
It can be seen that unless 
\begin{equation}
\label{const}
\rho-\lambda+\Phi^\dag_{\cal N}(H)=0,
\end{equation}
$g(\lambda,H)$ can be made to go to $-\infty$ by suitably choosing $\sigma$. So, we include $(\rho-\lambda+\Phi^\dag_{\cal N}(H))=0$ as a constraint. Combining this with the inequality constraints, one can write $\rho+\Phi^\dag_{QL}(H) \leq 0$  by eliminating $\lambda$, because for any $\lambda \geq 0$ a suitable $H$ satisfying Eq.~\eqref{const} will also satisfy this inequality. Also notice that $H$ appears in this dual optimization problem through $\Phi^\dag_{\cal N}(H)$. Hence, for any solution $H$, its QL projection $\Phi_{\cal N}(H)$ would work as well. In view of this, we may introduce an additional constraint,  
$\Phi^\dag_{\cal N}(H) = H, $
and simply replace $\Phi^\dag_{\cal N}(H)$ with $H$ everywhere else in the dual problem. Finally, the desired form for the dual problem is obtained, as given in Table I(b) in the main text:

\begin{center}
	\begin{tabular}{l l}
		{\tt Maximize:} & $-\text{tr}(H\rho)$, \\
		{\tt subject to:} & $H+\rho\geq 0$, \\
		& $H = \Phi_{\cal N}(H),$ \\
		& $H =H^\dag.$ 
	\end{tabular} 
\end{center}

Under the working assumptions of our problem, it is possible to show that {\em strong duality} holds, which means that the optimal value of dual problem is not just a lower bound for the primal (weak duality), but they are in fact {\em equal} under the active constraints. This can be shown by resorting to the refined version of Slater's condition for a SDP with affine inequality constraints (see e.g. \cite{Boyd}, Sec. 5.2.3): in this case, feasibility of the primal problem is enough to guarantee $\alpha=\beta.$

However, as we remark in Sec. \ref{sec:sdp}, it is important to notice that, unlike the primal problem, the dual problem requires optimization over an unbounded set. In this case, even with Slater's condition there is {\em no} guarantee that the optimal value is reached for a bounded solution. In fact, we have shown in Sec. \ref{sec:equivalence} that it is possible to have UDA states that are not UGS of any Hamiltonian respecting the same QL constraint -- which is equivalent to the dual problem having no bounded solution.

\end{appendix}


\end{document}